\documentclass[journal]{IEEEtran}

\usepackage{cite}
\usepackage{amsmath,amssymb,amsfonts,amsthm}
\usepackage{algorithmic}
\usepackage{graphicx,xcolor}
\usepackage{textcomp}
\usepackage{cases}
\usepackage[utf8]{inputenc}
\usepackage[english]{babel}
\def\BibTeX{{\rm B\kern-.05em{\sc i\kern-.025em b}\kern-.08em
    T\kern-.1667em\lower.7ex\hbox{E}\kern-.125emX}}

\usepackage{multirow}
\usepackage{pifont}

\usepackage{epstopdf}

\usepackage{enumitem}
\newtheorem{lemma}{Lemma}

\usepackage{hyperref}
\hypersetup{colorlinks=true,linkcolor=black,urlcolor=blue,}

\usepackage{dirtytalk}

\usepackage[caption=false]{subfig} 
\makeatletter
\renewcommand{\fnum@figure}{Fig. \thefigure}
\makeatother

\usepackage[acronym,toc,shortcuts]{glossaries}
\newacronym{Ph.D.}{Ph.D.}{Doctor of Philosophy}
\newacronym{SA}{SA}{Simulated Annealing}
\newacronym{VLC}{VLC}{visible light communications}
\newacronym{RF}{RF}{radio frequency}
\newacronym{V2V}{V2V}{vehicle-to-vehicle}
\newacronym{V2X}{V2X}{vehicle-to-everything}
\newacronym{V2I}{V2I}{vehicle-to-infrastructure}
\newacronym{B5G}{B5G}{beyond-fifth generation}
\newacronym{LED}{LED}{light emitting diode}
\newacronym{OMA}{OMA}{orthogonal multiple access}
\newacronym{FDMA}{FDMA}{frequency-division multiple-access}
\newacronym{TDMA}{TDMA}{time-division multiple-access}
\newacronym{CDMA}{CDMA}{code-division multiple-access}
\newacronym{OFDMA}{OFDMA}{orthogonal frequency-division multiple-access}
\newacronym{OFDM}{OFDM}{orthogonal frequency-division multiplexing}
\newacronym{WDMA}{WDMA}{wavelength-division multiple-access}
\newacronym{NOMA}{NOMA}{non-orthogonal multiple access}
\newacronym{PD-NOMA}{PD-NOMA}{power-domain NOMA}
\newacronym{CD-NOMA}{CD-NOMA}{code-domain NOMA}
\newacronym{SC}{SC}{superposition coding}
\newacronym{SIC}{SIC}{successive interference cancellation}
\newacronym{BS}{BS}{base station}
\newacronym{QoS}{QoS}{quality-of-service}
\newacronym{NP}{NP}{non-deterministic polynomial-time}
\newacronym{DCO-OFDM}{DCO-OFDM}{direct-current biased optical-OFDM}
\newacronym{DCO-OFDMA}{DCO-OFDMA}{direct-current biased optical-OFDMA}
\newacronym{DC}{DC}{direct current}
\newacronym{ITU}{ITU}{international telecommunication union}
\newacronym{FoV}{FoV}{field-of-view}
\newacronym{CSI}{CSI}{channel state information}
\newacronym{LACO-OFDM}{LACO-OFDM}{layered asymmetrically clipped optical OFDM}
\newacronym{ACO-OFDM}{ACO-OFDM}{asymmetrically clipped optical OFDM}
\newacronym{FR}{FR}{frequency reuse}
\newacronym{EAs}{EAs}{evolutionary algorithms}
\newacronym{C-LiAN}{C-LiAN}{centralized light access network}
\newacronym{AP}{AP}{access point}
\newacronym{PD}{PD}{photo-diode}
\newacronym{SINR}{SINR}{signal-to-noise-interference ratio}
\newacronym{LoS}{LoS}{line-of-sight}
\newacronym{AWGN}{AWGN}{additive white Gaussian noise}
\newacronym{SNR}{SNR}{signal-to-noise ratio}
\newacronym{NLUPA}{NLUPA}{next-largest-difference user-pairing algorithm}
\newacronym{D-NLUPA}{D-NLUPA}{divide-and-next-largest-difference user-pairing algorithm}
\newacronym{NAICS}{NAICS}{network-assisted interference cancellation and suppression}
\newacronym{LTE}{LTE}{long term evolution}
\newacronym{3GPP}{3GPP}{3rd Generation Partnership Project}
\newacronym{CR}{CR}{cognitive radio}
\newacronym{2D}{2D}{two-dimension}
\newacronym{GP}{GP}{gradient projection}
\newacronym{umMTC}{umMTC}{ultra-massive machine-type communication}
\newacronym{mMTC}{mMTC}{massive machine-type communication}
\newacronym{IoE}{IoE}{internet-of-everything}
\newacronym{IoUT}{IoUT}{internet-of-underwater-things}
\newacronym{ZF}{ZF}{zero-forcing}
\newacronym{NLIP}{NLIP}{non-linear integer programming}
\newacronym{DP}{DP}{dynamic programming}
\newacronym{MIMO}{MIMO}{multiple-input multiple-output}
\newacronym{TS}{TS}{Tabu-search}
\newacronym{THz}{THz}{Terahertz}
\newacronym{MISO}{MISO}{multiple-input single-output}
\newacronym{SIMO}{SIMO}{single-input multiple-output}
\newacronym{EE}{EE}{energy efficiency}
\newacronym{SEE}{SEE}{sum energy efficiency}
\newacronym{VR}{VR}{virtual reality}
\newacronym{XR}{XR}{extended reality}
\newacronym{5G}{5G}{fifth generation}
\newacronym{6G}{6G}{sixth generation}
\newacronym{NR}{NR}{new radio}
\newacronym{mmWave}{mmWave}{millimeter-wave}
\newacronym{FD}{FD}{full-duplex}
\newacronym{CNOMA}{CNOMA}{cooperative NOMA}
\newacronym{ABF}{ABF}{analog beamforming}
\newacronym{BF}{BF}{beamforming}
\newacronym{DF}{DF}{decode-and-forward}
\newacronym{AF}{AF}{amplify-and-forward}
\newacronym{CF}{CF}{compress-and-forward}
\newacronym{SPS}{SPS}{single-phase shifter}
\newacronym{PS}{PS}{phase shifter}
\newacronym{PA}{PA}{power amplifier}
\newacronym{NLoS}{NLoS}{non-line-of-sight}
\newacronym{ULA}{ULA}{uniform linear array}
\newacronym{SI}{SI}{self-interference}
\newacronym{MA}{MA}{multiple access}
\newacronym{1G}{1G}{first generation}
\newacronym{2G}{2G}{second generation}
\newacronym{3G}{3G}{third generation}
\newacronym{4G}{4G}{fourth generation}
\newacronym{M2M}{M2M}{machine-to-machine}
\newacronym{IoT}{IoT}{internet-of-things}
\newacronym{IMT}{IMT}{International Mobile Telecommunications}
\newacronym{SE}{SE}{spectral efficiency}
\newacronym{MMF}{MMF}{Maximin Fairness}
\newacronym{SR}{SR}{sum rate}
\newacronym{WSR}{WSR}{weighted sum rate}
\newacronym{D-NOMA}{D-NOMA}{dynamic-NOMA}
\newacronym{GSM}{GSM}{global system for mobile telecommunications}
\newacronym{IS-95}{IS-95}{Interim Standard 95}
\newacronym{SMS}{SMS}{short-message service}
\newacronym{CoMP}{CoMP}{coordinated multi-point}
\newacronym{MU-MIMO}{MU-MIMO}{multi-user MIMO}
\newacronym{MAC}{MAC}{multiple-access channel}
\newacronym{BC}{BC}{vector-broadcast channel}
\newacronym{CSIT}{CSIT}{channel state information at the transmitter}
\newacronym{DPC}{DPC}{dirty paper coding}
\newacronym{ZF-DPC}{ZF-DPC}{zero-forcing DPC}
\newacronym{BD}{BD}{block-diagonalization}
\newacronym{ICI}{ICI}{inter-cell interference}
\newacronym{MMSE}{MMSE}{minimum mean square error}
\newacronym{LTE-A}{LTE-A}{long term evolution-advanced}
\newacronym{MUST}{MUST}{multi-user superposition transmission}
\newacronym{SISO}{SISO}{single-input single-output}
\newacronym{LDS-CDMA}{LDS-CDMA}{low-density spreading CDMA}
\newacronym{LDS-OFDM}{LDS-OFDM}{low-density spreading OFDM}
\newacronym{SCMA}{SCMA}{sparse code multiple access}
\newacronym{MUSA}{MUSA}{multi-user sharing access}
\newacronym{SAMA}{SAMA}{successive interference cancellation amenable multiple access}
\newacronym{PDMA}{PDMA}{pattern division multiple access} 
\newacronym{BOMA}{BOMA}{building block sparse-constellation based orthogonal multiple access}  
\newacronym{LPMA}{LPMA}{lattice partition multiple access} \newacronym{OOK}{OOK}{on-off keying} 
\newacronym{M-PAM}{M-PAM}{M-ary pulse-amplitude modulation} 
\newacronym{M-PPM}{M-PPM}{M-ary pulse-position modulation} \newacronym{MSM}{MSM}{multiple-subcarrier modulation}
\newacronym{IM/DD}{IM/DD}{intensity modulation and direct detection}
\newacronym{RC}{RC}{repetition code}
\newacronym{SM}{SM}{spatial multiplexing}
\newacronym{SMOD}{SMOD}{spatial modulation}
\newacronym{BER}{BER}{bit error rate}
\newacronym{SER}{SER}{symbol error rate}
\newacronym{MFTP}{MFTP}{maximum flickering time period}
\newacronym{MU-MISO}{MU-MISO}{multi-user multi-input single-output}
\newacronym{MSE}{MSE}{minimum square error}
\newacronym{SPCA}{SPCA}{sequential parametric convex approximation}
\newacronym{WSMSE}{WSMSE}{weighted sum minimum square error}
\newacronym{OCDMA}{OCDMA}{optical code division multiple access}
\newacronym{SDMA}{SDMA}{space-division multiple access}
\newacronym{DMT}{DMT}{discrete multi-tone}
\newacronym{VLNs}{VLNs}{visible light networks}
\newacronym{VHO}{VHO}{Vertical handover}
\newacronym{RSS}{RSS}{received signal strength}
\newacronym{RSI}{RSI}{received signal intensity}
\newacronym{IA}{IA}{interference alignment}
\newacronym{BIA}{BIA}{blind interference alignment}
\newacronym{BBU}{BBU}{base-band unit}
\newacronym{SU}{SU}{secondary user}
\newacronym{PU}{PU}{primary user}
\newacronym{mMIMO}{mMIMO}{massive-MIMO}
\newacronym{UAV}{UAV}{unmanned aerial vehicle}
\newacronym{PHY}{PHY}{physical}
\newacronym{CF-mMIMO}{CF-mMIMO}{cell-free mMIMO}
\newacronym{LIS}{LIS}{large intelligent surfaces}
\newacronym{3-D MIMO}{3-D MIMO}{3-Dimensional MIMO}
\newacronym{RIS}{RIS}{reflecting intelligent surface}
\newacronym{BackCom}{BackCom}{backscatter communications}
\newacronym{UL}{UL}{uplink}
\newacronym{UE}{UE}{user equipment}
\newacronym{D2D}{D2D}{device-to-device}
\newacronym{FCC}{FCC}{Federal Communications Commission}
\newacronym{HAP}{HAP}{high altitude platform}
\newacronym{LAP}{LAP}{low altitude platform}
\newacronym{MEC}{MEC}{mobile edge computing}
\newacronym{NATO}{NATO}{North Atlantic Treaty Organization}
\newacronym{ML}{ML}{machine learning}
\newacronym{QML}{QML}{quantum machine learning}
\newacronym{DL}{DL}{deep learning}
\newacronym{DRL}{DRL}{deep Reinforcement learning}
\newacronym{RL}{RL}{Reinforcement learning}
\newacronym{MMA}{MMA}{minorization maximization algorithm}
\newacronym{MM}{MM}{majorization-minimization}
\newacronym{KKT}{KKT}{Karush–Kuhn–Tucker}
\newacronym{FDD}{FDD}{frequency division duplex}
\newacronym{SCA}{SCA}{sine-cosine algorithm}
\newacronym{AoD}{AoD}{angle-of-departure}
\newacronym{SDP}{SDP}{semi-definite programming}
\newacronym{SDR}{SDR}{semi-definite relaxation}
\newacronym{JT}{JT}{joint transmission}
\newacronym{CB}{CB}{coordinated beamforming}
\newacronym{RAMA}{RAMA}{relay-aided multiple access}
\newacronym{MRC}{MRC}{maximum ratio combining}
\newacronym{SWIPT}{SWIPT}{simultaneous wireless information and power transfer}
\newacronym{HetNets}{HetNets}{heterogeneous networks}
\newacronym{D.C.}{D.C.}{difference of convex}
\newacronym{GRPA}{GRPA}{gain ratio power allocation}
\newacronym{FPA}{FPA}{fixed power allocation}
\newacronym{CS}{CS}{Cuckoo Search}
\newacronym{HHO}{HHO}{Harris Hawks Optimizer}
\newacronym{PLC}{PLC}{power line communications}
\newacronym{HTT}{HTT}{harvest-then-transmit}
\newacronym{H-CRAN}{H-CRAN}{heterogeneous cloud radio access network}
\newacronym{RRHs}{RRHs}{remote radio heads}
\newacronym{IIoT}{IIoT}{Industrial IoT}
\newacronym{PAPR}{PAPR}{peak-to-average-power-ratio}
\newacronym{ANC}{ANC}{analog network coding}
\newacronym{FFR}{FFR}{fractional frequency reuse}
\newacronym{RGB}{RGB}{red-green-blue}
\newacronym{MAR}{MAR}{mobile augmented reality}
\newacronym{HD}{HD}{half-duplex}
\newacronym{CPU}{CPU}{central process unit}
\newacronym{RWP}{RWP}{Random Way-Point}
\newacronym{LC}{LC}{liquid crystal}
\newacronym{ADR}{ADR}{angle diversity receiver}
\newacronym{RSMA}{RSMA}{rate splitting multiple access}
\newacronym{omni-DRIS}{omni-DRIS}{omni-digital-RIS}
\newacronym{STAR-RIS}{STAR-RIS}{simultaneous transmission and reflection reconfigurable intelligent surface}
\newacronym{OSTAR-RIS}{OSTAR-RIS}{optical simultaneous transmission and reflection reconfigurable intelligent surface}
\glsdisablehyper

\usepackage[linesnumbered,ruled,vlined]{algorithm2e}

\begin{document}

\title{Computation Rate Maximization for Wireless Powered Edge Computing With Multi-User Cooperation}

\author{Yang Li, Xing Zhang,~\IEEEmembership{Senior Member,~IEEE,} Bo Lei, Qianying Zhao, Min Wei, Zheyan Qu, and Wenbo Wang,~\IEEEmembership{Senior Member,~IEEE}%

\thanks{Manuscript received 2 Jan 2024; accepted 20 Jan 2024. This work is supported by the National Science Foundation
of China under Grant 62071063, 62271062. (Corresponding author: Xing Zhang.)}
\thanks{Y. Li, X. Zhang, Z. Qu, and W. Wang are with the School of Information and Communications Engineering, Beijing University of Posts and Telecommunications, Beijing 100876, China (e-mail: ly209991@bupt.edu.cn; zhangx@ieee.org; zheyanqu@bupt.edu.cn; wbwang@bupt.edu.cn).}
\thanks{B. Lei, Q. Zhao, M. Wei are with Beijing Branch of China Telecom Co., Ltd., Beijing 100032, China (e-mail: leibo@chinatelecom.cn; zhaoqy50@chinatelecom.cn; weim6@chinatelecom.cn).}
\thanks{$\copyright$ 2024 IEEE. Personal use of this material is permitted. Permission from IEEE must be obtained for all other uses, in any current or future media, including reprinting/republishing this material for advertising or promotional purposes, creating new collective works, for resale or redistribution to servers or lists, or reuse of any copyrighted component of this work in other works.}%
\thanks{Digital Object Identifier:\href{https://ieeexplore.ieee.org/abstract/document/12345}{10.1109/OJCOMS.2024.123456}}
}

\markboth{Accepted for publication in IEEE Open Journal of the Communications Society, Jan. 2024.}%
{Yang Li \MakeLowercase{\textit{et al.}}: Computation Rate Maximization for Wireless Powered Edge Computing With Multi-User Cooperation}

\maketitle

\begin{abstract}
The combination of mobile edge computing (MEC) and radio frequency-based wireless power transfer (WPT) presents a promising technique for providing sustainable energy supply and computing services at the network edge. This study considers a wireless-powered mobile edge computing system that includes a hybrid access point (HAP) equipped with a computing unit and multiple Internet of Things (IoT) devices. In particular, we propose a novel muti-user cooperation scheme to improve computation performance, where collaborative clusters are dynamically formed. Each collaborative cluster comprises a source device (SD) and an auxiliary device (AD), where the SD can partition the computation task into various segments for local processing, offloading to the HAP, and remote execution by the AD with the assistance of the HAP. Specifically, we aims to maximize the weighted sum computation rate (WSCR) of all the IoT devices in the network. This involves jointly optimizing collaboration, time and data allocation among multiple IoT devices and the HAP, while considering the energy causality property and the minimum data processing requirement of each device. Initially, an optimization algorithm based on the interior-point method is designed for time and data allocation. Subsequently, a priority-based iterative algorithm is developed to search for a near-optimal solution to the multi-user collaboration scheme. Finally, a deep learning-based approach is devised to further accelerate the algorithm’s operation, building upon the initial two algorithms. Simulation results show that the performance of the proposed algorithms is comparable to that of the exhaustive search method, and the deep learning-based algorithm significantly reduces the execution time of the algorithm.     
\end{abstract}


\maketitle

\section{Introduction}\label{I}
\subsection{Background}\label{I-A}
\IEEEPARstart{T}{he} Industrial Internet of Things (IIoT) is an application of Internet of Things (IoT) in the industrial sector \cite{1}. It represents a new stage for the development of intelligence in traditional industries. IoT nodes with heterogeneous sensors are deployed in an IIoT network to perform various monitoring and control functions. The data collected by IoT nodes can be analyzed to enhance production effectiveness and reliability \cite{2}. Generally, data is collected through IoT nodes' access points and transmitted to a cloud computing center for processing and analysis. However, cloud computing introduces long transmission delays, rendering it unsuitable for real-time IIoT scenarios. Moreover, IIoT systems place a strong emphasis on the protection of private data, especially when dealing with confidential products that may involve sensitive information \cite{3}.

To enhance data processing and privacy protection in IIoT systems, empowering access points with mobile edge computing (MEC) capability is under consideration \cite{4,5}. MEC enables low-latency operations by extending cloud computing and services to the network edge \cite{42,43,44,45,46}. It also enhances the ability to protect private data, as data processing occurs exclusively within the plant's network. However, the computing capabilities of the access point and IoT nodes are limited. Additionally, data collected by different IoT nodes has varying levels of importance. Therefore, we aim to analyze as much high-importance data as possible while ensuring that each IoT node meets the minimum data processing requirement. This goal enables personalized and targeted data analysis in the IIoT system.
 
In the context of MEC, the literature extensively investigates two popular offloading policies \cite{6,27,13,29,30}: partial offloading and binary offloading. The former is applicable to arbitrarily divisible computing tasks. In this case, each task is divided into two parts: one is offloaded to the MEC server, and the other is computed locally by the IoT device. The latter is suitable for indivisible computing tasks, such as simple atomic tasks. Therefore, each task is computed either by the IoT node or the MEC server. However, in the scenario we consider, these two strategies fail to harness the IIoT system's computational capacity fully. This study explores collaboration among multiple IoT nodes and the hybrid access point (HAP) to maximize the weighted sum computation rate (WSCR) of the IIoT system.

On the other hand, due to the stringent device size constraint and production cost consideration, an IoT device often carries a capacity-limited battery, which needs frequent battery replacement and charging. However, this requirement can result in significant labor costs and financial burdens for service providers, particularly in large-scale IIoT systems and harsh environments \cite{8,9}. Therefore, radio frequency (RF)-based wireless power transfer (WPT) has emerged as an effective solution to the problem of finite battery capacity \cite{10}. Specifically, WPT employs a dedicated RF energy transmitter to charge the batteries of remote energy-harvesting devices continuously \cite{11}. According to reports, commercial WPT transmitters can continuously supply RF power of approximately 5W over distances exceeding 10 meters \cite{12}. Deploying the information transceiver, energy transmitter, and computing unit on the HAP to provide both MEC services and WPT can yield their benefits simultaneously. Therefore, wireless-powered mobile edge computing has garnered significant interest as a viable solution to address computing capacity constraints and energy consumption issues in IIoT systems.
\subsection{Related Work}\label{I-B}
\begin{table*}[]
\centering
\caption{Overview of the literature related to wireless-powered MEC networks.
}
\label{tab1}
\Large
\resizebox{\textwidth}{!}{%
\begin{tabular}{llllll}
\hline
\textbf{Research}                                                               & \textbf{Number of users} & \textbf{\begin{tabular}[c]{@{}l@{}}Collaboration mode\\ among users\end{tabular}}                                                                                               & \textbf{Objective}                                                                                         & \textbf{Technical method}                                                                                                          & \textbf{Contributions}                                                                                                                                                                      \\ \hline
\begin{tabular}[c]{@{}l@{}}F. Wang et al. \\ (2020) \cite{15}\end{tabular}      & Single user              & No collaboration                                                                                                                                                                & \begin{tabular}[c]{@{}l@{}}Minimize the total \\ transmitted energy \\ consumption of the ET.\end{tabular} & Heuristic algorithm                                                                                                                & \begin{tabular}[c]{@{}l@{}}The energy and task allocation are \\ jointly optimized.\end{tabular}                                                                                            \\ \hline
\begin{tabular}[c]{@{}l@{}}C. Psomas et al. \\ (2020) \cite{16}\end{tabular}    & Single user              & No collaboration                                                                                                                                                                & \begin{tabular}[c]{@{}l@{}}Maximize the system's \\ utility.\end{tabular}                                  & \begin{tabular}[c]{@{}l@{}}Cumulative distribution \\ functions (CDFs) and \\ probability density \\ functions (PDFs)\end{tabular} & \begin{tabular}[c]{@{}l@{}}Analytical expressions for the \\ probability of successful computa-\\ tion and the average number of \\ successfullycomputed bits are \\ provided.\end{tabular} \\ \hline
\begin{tabular}[c]{@{}l@{}}H. Li et al. \\ (2023) \cite{17}\end{tabular}        & Multiple users           & No collaboration                                                                                                                                                                & \begin{tabular}[c]{@{}l@{}}Minimize the transmit \\ power of HAP.\end{tabular}                             & \begin{tabular}[c]{@{}l@{}}Alternating direction \\ method of multipliers \\ (ADMM)\end{tabular}                                   & \begin{tabular}[c]{@{}l@{}}The time allocation and computa-\\ tional mode selection are jointly\\ optimized.\end{tabular}                                                                   \\ \hline
\begin{tabular}[c]{@{}l@{}}S. Bi et al.\\ (2018) \cite{31}\end{tabular}         & Multiple users           & No collaboration                                                                                                                                                                & \begin{tabular}[c]{@{}l@{}}Maximize the weighted\\ sum computation rate.\end{tabular}                      & ADMM                                                                                                                               & \begin{tabular}[c]{@{}l@{}}The individual computing mode sele-\\ ction (binary offloading) and time \\ allocation are jointly optimized.\end{tabular}                                       \\ \hline
\begin{tabular}[c]{@{}l@{}}J. Liu et al. \\ (2020) \cite{18}\end{tabular}       & Multiple users           & No collaboration                                                                                                                                                                & \begin{tabular}[c]{@{}l@{}}Maximize the minimum \\ energy balance among \\ all users.\end{tabular}         & \begin{tabular}[c]{@{}l@{}}Generalized Benders \\ decomposition (GBD)\end{tabular}                                                 & \begin{tabular}[c]{@{}l@{}}Time assignments, CPU frequencies,\\ and the computing mode selection \\ are jointly optimized.\end{tabular}                                                     \\ \hline
\begin{tabular}[c]{@{}l@{}}K. Xiong et al. \\ (2022) \cite{19}\end{tabular}     & Multiple users           & No collaboration                                                                                                                                                                & \begin{tabular}[c]{@{}l@{}}Minimize the energy \\ consumption of the \\ UAV.\end{tabular}                  & \begin{tabular}[c]{@{}l@{}}Successive convex \\ approximation (SCA)\end{tabular}                                                   & \begin{tabular}[c]{@{}l@{}}The UAV trajectory, the offloading \\ amount, and the computational res-\\ ource allocation are jointly opt-\\ imized.\end{tabular}                              \\ \hline
\begin{tabular}[c]{@{}l@{}}P. X. Nguyen et al. \\ (2021) \cite{20}\end{tabular} & Multiple users           & No collaboration                                                                                                                                                                & \begin{tabular}[c]{@{}l@{}}Maximize the weighted \\ sum computation rate.\end{tabular}                     & \begin{tabular}[c]{@{}l@{}}Block coordinate \\ descent\end{tabular}                                                                & \begin{tabular}[c]{@{}l@{}}The transmit power, backscatter \\ coefficient, time-splitting ratio, \\ and binary decision-making matri-\\ ces are jointly optimized.\end{tabular}             \\ \hline
\begin{tabular}[c]{@{}l@{}}F. Zhou et al. \\ (2020) \cite{21}\end{tabular}      & Multiple users           & No collaboration                                                                                                                                                                & \begin{tabular}[c]{@{}l@{}}Maximize the computa-\\ tional efficiency.\end{tabular}                         & SCA                                                                                                                                & \begin{tabular}[c]{@{}l@{}}The time allocation, the local \\ computing frequency, and power \\ are jointly optimized.\end{tabular}                                                          \\ \hline
\begin{tabular}[c]{@{}l@{}}B. Su et al. \\ (2021) \cite{22}\end{tabular}        & Two users                & \begin{tabular}[c]{@{}l@{}}The near user assists\\ the far user.\end{tabular}                                                                                                   & \begin{tabular}[c]{@{}l@{}}Maximize the computa-\\ tional efficiency.\end{tabular}                         & \begin{tabular}[c]{@{}l@{}}Semidefinite relaxation \\ (SDR) and SCA\end{tabular}                                                   & \begin{tabular}[c]{@{}l@{}}The energy beamforming, time, \\ and power allocation are jointly\\ designed.\end{tabular}                                                                       \\ \hline
\begin{tabular}[c]{@{}l@{}}B. Li et al. \\ (2021) \cite{23}\end{tabular}        & Two users                & \begin{tabular}[c]{@{}l@{}}The near user assists\\ the far user.\end{tabular}                                                                                                   & \begin{tabular}[c]{@{}l@{}}Minimize the energy \\ consumption.\end{tabular}                                & Lagrange duality method                                                                                                            & \begin{tabular}[c]{@{}l@{}}The offloading time, task all-\\ ocation and computation frequency\\ are jointly optimized.\end{tabular}                                                         \\ \hline
\begin{tabular}[c]{@{}l@{}}X. Wu et al. \\ (2023) \cite{24}\end{tabular}        & Multiple users           & \begin{tabular}[c]{@{}l@{}}A helper assists a\\ source node (Coll-\\ aboration relation-\\ ships among users \\ are fixed).\end{tabular}                                        & \begin{tabular}[c]{@{}l@{}}Maximize the weighted\\ sum computation rate.\end{tabular}                      & Penalty function method                                                                                                            & \begin{tabular}[c]{@{}l@{}}The time assignment, computation-\\ task allocation, and transmission\\ power are jointly optimized.\end{tabular}                                                \\ \hline
Our study                                                                       & Multiple users           & \begin{tabular}[c]{@{}l@{}}An auxiliary device\\ assists a source node\\ (Collaborative rela-\\ tionships dynamically \\ change across differ-\\ ent time frames).\end{tabular} & \begin{tabular}[c]{@{}l@{}}Maximize the weighted\\ sum computation rate.\end{tabular}                      & \begin{tabular}[c]{@{}l@{}}Interior-point method, \\ priority-based iterative \\ algorithm and deep \\ learning\end{tabular}       & \begin{tabular}[c]{@{}l@{}}The collaboration relationships,\\ time and data allocation among \\ multiple IoT devices and the HAP\\ are jointly optimized.\end{tabular}                      \\ \hline
\end{tabular}%
}
\end{table*}

In the existing literature \cite{13,14,15,16,17,18,19,20,21,22,31,23,24}, MEC and WPT have been extensively investigated in IoT systems. In \cite{13}, an MEC network was studied to minimize the execution delay of total tasks. The study proposed an optimized offloading framework that allows multiple IoT devices to offload tasks to multiple edge computing servers. In \cite{14}, the authors focused on minimizing the overall energy consumption of IoT devices. The research investigated resource allocation techniques for multi-user MEC systems based on time division multiple access (TDMA) and orthogonal frequency division multiple access (OFDMA). However, the WPT was not involved in \cite{13} and \cite{14}.

The optimal design of wireless-powered MEC networks differs significantly from traditional non-WPT-involved MEC networks due to the trade-off between WPT and wireless information transfer (WIT) \cite{17}. To effectively leverage the benefits of MEC and RF-based WPT, several recent studies have delved into the investigation of wireless-powered MEC networks. In \cite{15}, a single-user wireless-powered MEC system was studied to minimize the total transmitted energy consumption of the energy transmitter (ET). The study developed heuristic online designs for joint energy and task allocation. In \cite{16}, the authors presented analytical expressions for the probability of successful computation and the average number of successfully computed bits in a single-user wireless-powered MEC system, aiming to maximize the system's utility. However, \cite{15} and \cite{16} only considered the single-user network model and the derived results may not be readily applicable to multi-user scenarios.

In practice, IIoT systems commonly consist of a large number of IoT devices. As a result, several studies have explored multi-user wireless-powered MEC networks. However, most of these works primarily focus on the trade-off between computing offloading and wireless charging, neglecting the aspect of collaboration among users. The study in \cite{31} investigated a multi-user MEC network powered by WPT. In this network, each energy-harvesting node adheres to a binary computation offloading policy. A joint optimization method, utilizing the alternating direction method of multipliers (ADMM) decomposition technique, was proposed. This method aims to maximize the weighted sum computation rate of all nodes in the network by optimizing both individual computing mode selection and system transmission time allocation jointly. Work in \cite{17} studied a wireless-powered multi-user MEC network in which the transmit power of HAP was minimized. The authors devised an algorithm based on the ADMM to optimize the time allocation and computational mode selection jointly. In \cite{18}, a wireless-powered multi-user hierarchical fog cloud computing network was studied to maximize the minimum energy balance among all users. To achieve the global optimal solution, the authors proposed a generalized Benders decomposition (GBD)-based approach. In \cite{19}, a rotary-wing unmanned aerial vehicle (UAV)-assisted multi-user wireless-powered MEC network was studied, aiming to minimize the energy consumption of the UAV. The authors devised an iterative method based on successive convex approximation (SCA) theory to solve the non-convex problem, jointly optimizing the UAV trajectory, the offloading amount, and the computational resource allocation. In \cite{20}, an OFDMA-based multi-user wireless-powered MEC system was studied to maximize the weighted sum computation rate. The study proposed an efficient algorithm based on block coordinate descent to achieve near-optimal performance. In \cite{21}, a wireless-powered MEC network with multiple users was studied to maximize computational efficiency. The research considered and evaluated the offloading based on both TDMA and non-orthogonal multiple access (NOMA) and proposed two iterative algorithms and two alternative optimization algorithms, respectively.

The diversity among different IoT devices has led to considering user collaboration (UC) as an effective approach to enhance the system's computational capacity \cite{24}. Consequently, numerous recent studies have explored the design of user collaboration in multi-user wireless-powered MEC networks. In \cite{22}, the authors investigated the application of UC and NOMA schemes in a two-user wireless-powered MEC network. They aimed to maximize computational efficiency by jointly designing energy beamforming, time, and power allocation. In \cite{23}, the authors proposed a NOMA-assisted UC scheme in a three-node WPT-MEC system. The system comprises a far user (FU), a near user (NU), and an access point (AP). The FU employs NOMA, enabling the simultaneous offloading of data to the NU and the AP, thereby overcoming the dual far-near effect of the FU. However, in \cite{22} and \cite{23}, only two-user network models were considered, precisely the NU-assisted FU scenario. The derived results may face challenges when extending to multi-user scenarios. In \cite{24}, a multi-user WPT-MEC system was studied. The authors proposed a multi-user collaboration scheme based on OFDMA to enhance computational efficiency. The scheme allows users to divide computational tasks into multiple parts for local computation, offloading to corresponding assistants, and the HAP for remote execution. However, the collaborative relationships among users are fixed in this work. In fact, the optimal collaborative relationships among users may change as the state of the network changes.

Table \ref{tab1} summarizes the closely related works in wireless-powered MEC networks. A comparison between these works and ours primarily focuses on the number of users, collaboration mode, objective, technical method, and contributions.

\subsection{Contributions}\label{I-C}
The main contributions of this paper are summarized as follows.
\begin{itemize}
\item \emph{Firstly,} to maximize the WSCR of the system, we formulate a mixed-integer programming (MIP) problem for the multi-user wireless-powered MEC network. The problem jointly optimizes the collaboration strategy, time allocation, and data assignment among users and the HAP, subject to energy causality constraints and minimum processing data volume requirements.
\item \emph{Secondly,} we decompose the non-convex optimization problem into the time and data allocation sub-problem and the collaboration strategy design sub-problem. An interior-point method based algorithm is developed to solve the time and data allocation sub-problem given a fixed collaboration strategy.
\item \emph{Thirdly,} we propose a low-complexity priority-based iterative search algorithm to determine the multi-user collaboration strategy. This allows for finding a sub-optimal solution within a reduced search space.
\item \emph{Finally,} to further reduce algorithm execution time, we devise a deep learning-based algorithm building on the first two algorithms. Simulation results are performed to demonstrate the validity of the proposed algorithms.
\end{itemize}

It is worth noting that smart energy harvesting and edge computing are crucial technologies and challenges in the context of 6G \cite{40}. Meanwhile, in the realm of 6G, the IIoT stands out as a typical application scenario within IoT. Therefore, the investigation of wireless-powered mobile edge computing within the context of IIoT holds significant importance.

\subsection{Organization}\label{I-D}
The organization of our work is demonstrated as follows. The system model is introduced in Section II, and the weighted sum computation rate maximization problem is formulated in Section III. In Section IV, three algorithms are proposed to solve the formulated problem. Numerical simulation results are provided to validate the proposed algorithms in Section V. Finally, this paper is concluded in Section VI. The code and dataset are available at \href{https://github.com/CPNGroup/WPTMEC}{https://github.com/CPNGroup/WPTMEC}.

\section{System Model}

\begin{table}[t]
\centering
\caption{Summary of system notations}
\label{notation}
\resizebox{0.9\columnwidth}{!}{%
\begin{tabular}{ll}
\hline
Notation           & Description                                                                                                                                    \\ \hline
$\mathcal{N}$      & The set of IoT devices                                                                                                                    \\
$\boldsymbol{o}$   & The set of SDs                                                                                                                                 \\
$m$                & The number of collaborative clusters                                                                                                           \\
$\boldsymbol{p}$   & The set of ADs                                                                                                                                 \\
$\boldsymbol{q}$   & The set of IDs                                                                                                                                 \\
$T$                & The length of a time frame                                                                                                                     \\
$\alpha_1$         & \begin{tabular}[c]{@{}l@{}}The fraction of time that the HAP broadcasts RF \\ energy for the IoT devices to harvest\end{tabular}             \\
$\alpha_2$         & \begin{tabular}[c]{@{}l@{}}The fraction of time that the HAP receives data \\ offloaded by the IoT devices\end{tabular} \\
$\eta$             & The energy harvesting efficiency                                                                                                 \\
$p_n$              & The transmit power of the $n$th IoT device                                                                                                     \\
$P$                & The transmit power of the HAP                                                                                                                  \\
$h_n$              & \begin{tabular}[c]{@{}l@{}}The wireless channel gain between the HAP and the \\ $n$th device\end{tabular}                                      \\
$E_n^h$            & The amount of energy that the $n$th device harvests                                                                                            \\
$t_n^{loc}$        & The local computation time of the $n$th device                                                                                                 \\
$l_n^{loc}$        & \begin{tabular}[c]{@{}l@{}}The amount of data allocated for local processing \\ by the $n$th IoT device\end{tabular}                           \\
$\phi_n$             & \begin{tabular}[c]{@{}l@{}}The number of cycles needed to process one bit of \\ task data from the $n$th IoT device\end{tabular}                                         \\
$f_n$              & \begin{tabular}[c]{@{}l@{}}The processor’s computing speed (cycles per second) \\ of the $n$th IoT device\end{tabular}                         \\
$k_n$              & The computation energy efficiency coefficient                                                                                                  \\
$E_n^{loc}$        & \begin{tabular}[c]{@{}l@{}}The energy consumption of the $n$th device due to \\ the local computing\end{tabular}                               \\
$l_n^{ap}$         & \begin{tabular}[c]{@{}l@{}}The amount of data allocated for processing on the \\ HAP by the $n$th IoT device\end{tabular}                      \\
$l_{o_i}^{p_i}$    & \begin{tabular}[c]{@{}l@{}}The amount of data allocated for processing on the \\ AD $p_i$ by the SD $o_i$\end{tabular}                         \\
$t_n^{off}$        & \begin{tabular}[c]{@{}l@{}}The data transfer time for the $n$th IoT device \\ during computing offloading\end{tabular}                         \\
$E_n^{off}$        & \begin{tabular}[c]{@{}l@{}}The energy consumption of the $n$th device due to \\ computing offloading\end{tabular}                              \\
$l_{th}$           & \begin{tabular}[c]{@{}l@{}}The minimum data processing requirement of each \\ device within a time frame\end{tabular}                                 \\
$f_{ap}$           & \begin{tabular}[c]{@{}l@{}}The processor’s computing speed (cycles per \\ second) of the HAP\end{tabular}                                      \\
$\boldsymbol{l_o}$ & The profile of data distribution scheme for SDs                                                                                                \\
$\boldsymbol{l_p}$ & The profile of data distribution scheme for ADs                                                                                                \\
$\boldsymbol{l_q}$ & The profile of data distribution scheme for IDs                                                                                                \\ \hline
\end{tabular}%
}
\vspace{-0.5cm}
\end{table}

\subsection{Network Model}

In this paper, we consider a multi-user wireless-powered MEC network. The network consists of a HAP integrated with a computing unit and a set of single-antenna devices, denoted as $\mathcal{N}=\{1,..., N\}$. The HAP with $N$ antennas provides stable wireless energy supply and enables task offloading for the $N$ devices, where each device can convert the received RF signals into power for subsequent data offloading and processing. The system time is segmented into consecutive time frames, each with a uniform duration of $T$, deliberately chosen to be smaller than the channel coherence time. For instance, this could be on the order of several seconds in a static IoT environment \cite{39}. At the beginning of each time frame, the devices are categorized into three types based on the diversity in data importance, computational capabilities, and channel qualities. The three classes are source devices (SDs), auxiliary devices (ADs), and independent devices (IDs), as illustrated in Fig. \ref{system}. Let $\boldsymbol{o}=\{o_1,o_2,...,o_m\}$, $\boldsymbol{p}=\{p_1,p_2,...,p_m\}$, and $\boldsymbol{q}=\{q_1,q_2,...,q_{N-2m}\}$ denote the sets of SDs, ADs, and IDs, respectively, where $2m \leq N$ and $\boldsymbol{o}\cup\boldsymbol{p}\cup\boldsymbol{q}=\mathcal{N}$. Each pair $\{o_i,p_i\}$, $\forall i\in \{1,2,...,m\}$, forms a collaborative cluster, where $m$ is the number of collaborative clusters. As shown in Fig. \ref{data}, the data of each source device (SD) is partitioned into three parts for processing at the local device, HAP, and corresponding auxiliary device (AD), respectively. Specifically, for the portion of an SD's data that is designated for processing at the corresponding AD, it will first be transmitted to the HAP and then forwarded to the AD by the HAP. This avoids communication issues due to distance limitations between the SD and AD. The data of each AD is processed locally. The data of each independent device (ID) is partitioned into two parts for processing at the local device and HAP, respectively. The main symbols and notations are summarized in Table \ref{notation}.

\begin{figure}[t]
\centerline{\includegraphics[width=0.45\textwidth]{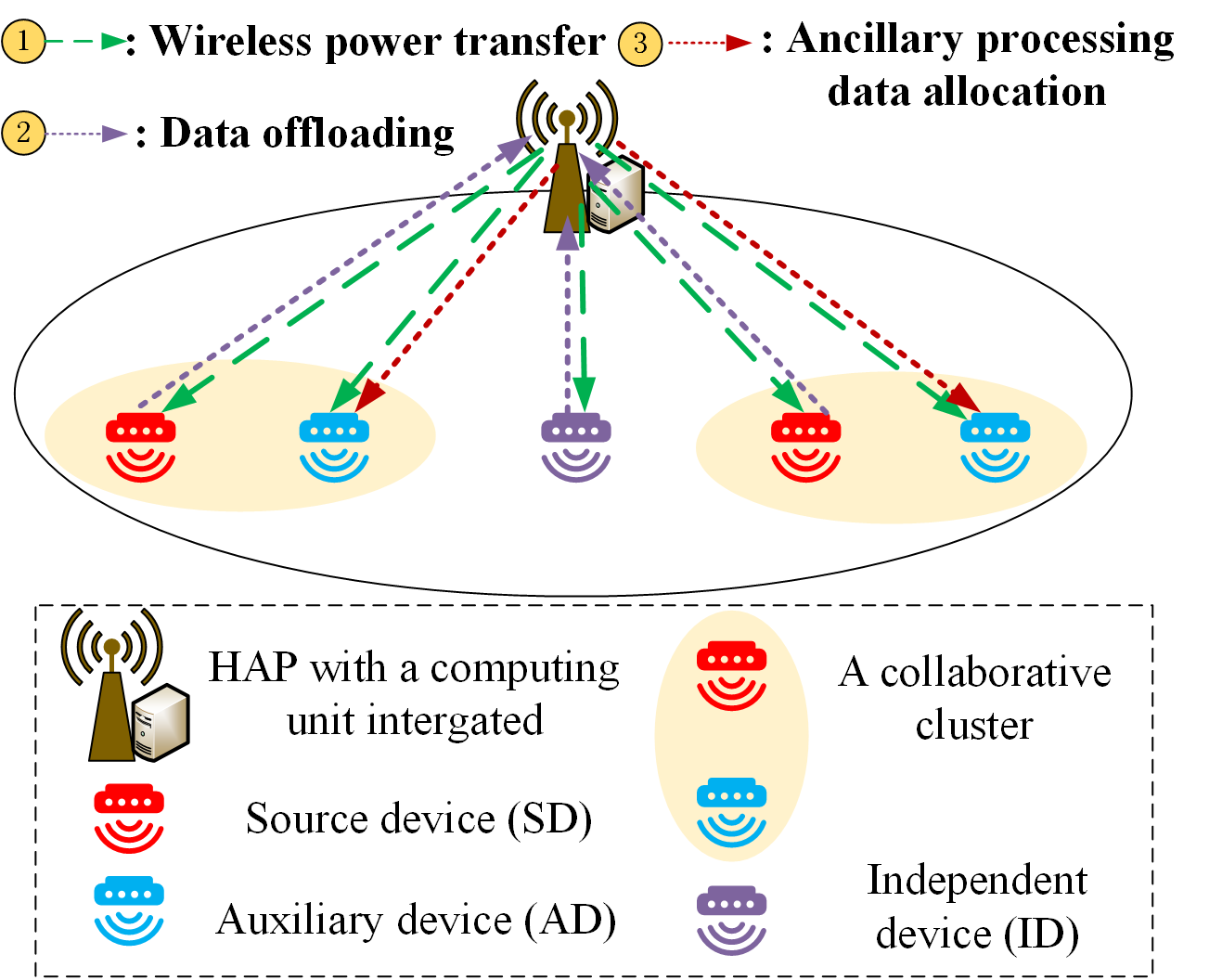}}
\caption{The wireless-powered MEC network with multi-user cooperation.}
\label{system}
\end{figure}

\begin{figure}[t]
\centerline{\includegraphics[width=0.45\textwidth]{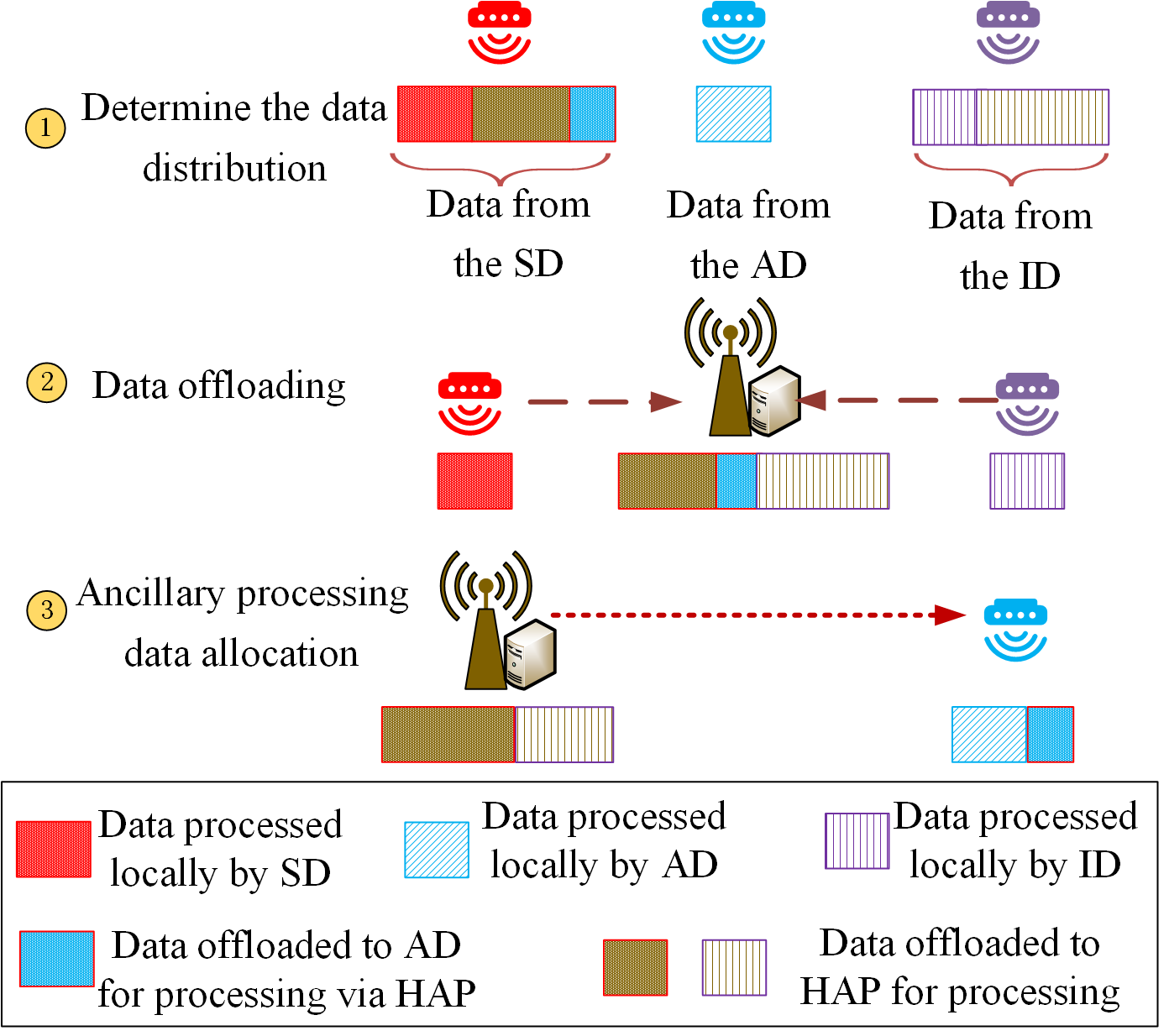}}
\caption{The data allocation for three types of devices. (The SD and AD form a collaborative cluster.)}
\label{data}
\end{figure}

\begin{figure*}[]
\centerline{\includegraphics[width=\textwidth]{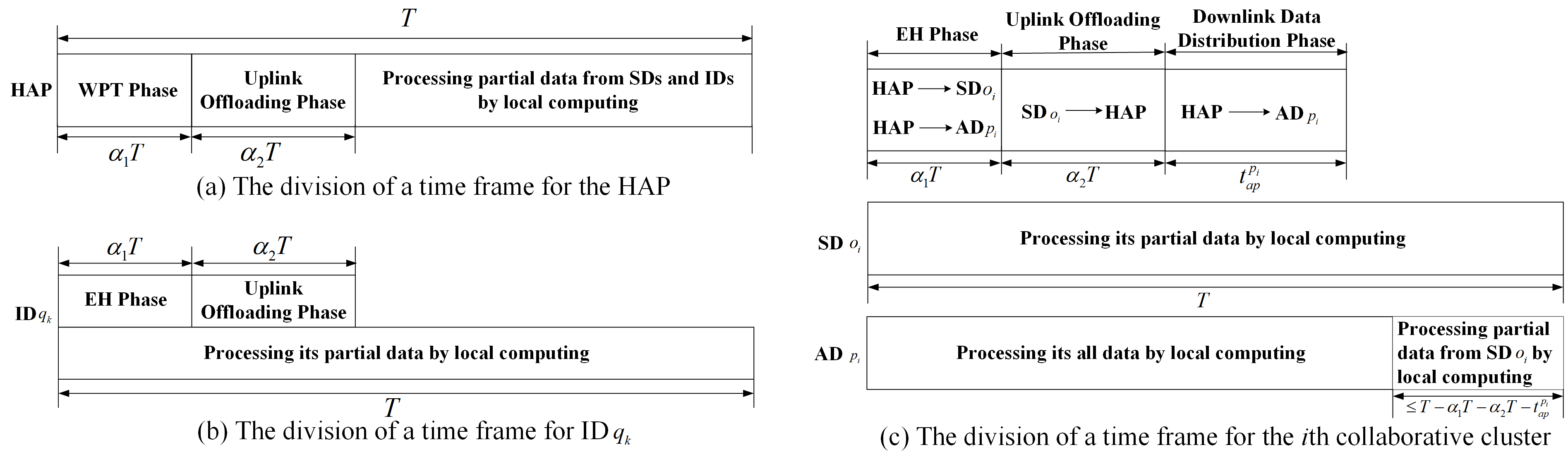}}
\caption{The harvest-and-offloading protocol with UC and OFDMA.}
\label{timeframe}
\end{figure*}


\subsection{Harvest-and-Offloading Protocol}

To avoid co-channel interference and facilitate collaboration among users as well as between users and HAP, we propose a harvest-and-offload protocol utilizing UC and OFDMA, as depicted in Fig. \ref{timeframe}. This protocol achieves collaboration among users as well as between users and HAP by planning the time allocation for different categories of IoT devices and HAP in each time frame. It is assumed that the same frequency band is utilized for both RF energy signals and data transmission. Moreover, based on OFDMA, the system bandwidth is divided into $N$ sub-channels. The HAP randomly assigns each device to a sub-channel to avoid inter-device interference. Consistent with \cite{20}, we adopt a quasistatic scenario where the wireless channel remains static within each time block but varies across blocks.

As shown in Fig. \ref{timeframe}(a), the time block of the HAP is divided into three stages. The first stage, with a duration of $\alpha_1T$ where $\alpha_1 \in (0,1]$, is allocated for WPT. The second stage, lasting for $\alpha_2T$ where $\alpha_2 \in [0,1)$, involves the HAP receiving offloaded data from all devices, including data designated for processing at the HAP and ADs. During the remaining time of $T-(\alpha_1+\alpha_2)T$, the HAP processes data received from ADs and IDs, which is intended for processing at the HAP.

As illustrated in Fig. \ref{timeframe}(b), IDs can locally process a portion of their data throughout the entire time block. Concurrently, during the initial $\alpha_1T$ of each time frame, it is dedicated to energy harvesting (EH). Subsequently, in the following $\alpha_2T$, IDs offload the data prepared for processing at the HAP through uplinks. It is noteworthy that the time for data offloading from each ID must be less than or equal to $\alpha_2T$.

Fig. \ref{timeframe}(c) illustrates the time allocation and interaction process between two devices in a cooperative cluster. In the initial $\alpha_1T$ of each time frame, both the SD $o_i$ and AD $p_i$ undergo EH. In the subsequent $\alpha_2 T$, SD $o_i$ offloads the data prepared for processing at both the HAP and AD $p_i$ to the HAP via the uplink. Notably, the actual data offloading time from SD $o_i$ to the HAP must be less than or equal to $\alpha_2 T$. However, it is only after $\alpha_1T+\alpha_2T$ that the HAP begins transmitting data to ADs, which is designated for processing by ADs and originates from SDs.

The time $t_{ap}^{p_i}$ for HAP to send data to AD $p_i$ through the downlink depends on their transmission rate and the amount of data that SD $o_i$ allocates for processing by AD $p_i$. Since AD $p_i$ can only start assisting processing after receiving data from SD $o_i$, the assisting processing time of AD $p_i$ cannot exceed $T-(\alpha_1+\alpha_2)T-t_{ap}^{p_i}$. Before the assisting processing, AD $p_i$ needs to process all its data locally. Additionally, throughout the entire time block, SD $o_i$ can process a portion of its data locally. It is worth noting that the delay for results collection is ignored, as the size of the results is negligible compared to the input data size, as stated in \cite{36,37}.

The proposed protocol utilizes OFDMA to reduce the possibility of co-channel interference among devices. It employs Time Division Multiplexing (TDD) on sub-channels, enabling energy harvesting and data transmission to share the same frequency band. This approach enhances both spectral and energy efficiency. Simultaneously, the protocol promotes efficient collaboration among various system components. It achieves this by strategically allocating time among different devices, ensuring sufficient time for energy harvesting, data transmission, and processing within each timeframe. Nevertheless, the practical implementation of this protocol faces potential challenges. These include communication quality fluctuations caused by dynamic environmental factors, the influence of collaboration delays on real-time performance, and variations in energy harvesting efficiency due to environmental changes and device mobility. Due to space limitations, a detailed discussion of these challenges will not be pursued.

\subsection{Wireless Energy Transmission}
At the beginning of each time frame, the HAP broadcasts wireless power to each user for a duration of $\alpha_1T$. Therefore, the harvested energy by device $n$ can then be expressed as
\begin{align}
\label{harvested energy}
\left.E_n^h=\left\{\begin{array}{ll}\eta P h_n \alpha_1 T,&\;\eta P h_n \alpha_1 T\leq E_n^{max},\\E_n^{max},&\;\eta P h_n \alpha_1 T>E_n^{max},\end{array}\right.\right. n\in \mathcal{N},
\end{align}
where constant $\eta \in (0,1)$ denotes the EH efficiency factor and $P$ represents the transmit power at the HAP. $h_n$ represents the wireless channel gain between the $n$th IoT device and the HAP, assuming it is the same for both uplink and downlink due to channel reciprocity \cite{31,32}. $E_n^{max}$ denotes the maximum battery capacity of the $n$th IoT device.

It should be noted that while some existing works have discussed the non-linear EH model and argued that it is closer to real-world scenarios \cite{33}, numerous studies (e.g., \cite{34,35}) have reported that EH circuits typically operate within the linear region when the distance between the HAP and IoT devices is not very short. This is due to the significant influence of path loss in WPT systems on the amount of harvested energy. Additionally, the linear model is frequently employed as they are tractable and accurate for moderate harvested power \cite{8}. Thus, similar to previous studies \cite{8, 31, 34, 35}, we also adopt the linear EH model when $E_n^h \le E_n^{max}$ as shown in (\ref{harvested energy}) in this paper.

\subsection{Data processing and constraints}
\subsubsection{Source Devices}
The data of SD $o_i$ is partitioned into three components: local processing, processing by the HAP, and processing by the corresponding AD $p_i$. The data processed by AD $p_i$ needs to be offloaded to the HAP first and then transmitted from the HAP to AD $p_i$ for processing. Regarding energy consumption, the SD $o_i$’s energy expenditure can be divided into two main parts: local processing and computing offloading. These two aspects are analyzed as follows.
\begin{itemize}
\item \emph{Local processing:} An IoT device can harvest energy and process its data simultaneously \cite{25},\cite{38}. Let $f_{o_i}$ denote the processor's computational speed (cycles per second) of SD $o_i$, and let $l_{o_i}^{loc}$ denote the amount of data processed locally. Therefore, the time required for local processing is given by $t_{o_i}^{loc}=\frac{l_{o_i}^{loc}\phi_{o_i}}{f_{o_i}}$, where $\phi_{o_i}$ represents the number of cycles needed to process one bit of task data from SD $o_i$. Moreover, the local processing must be completed before the time frame ends, which can be expressed as $t_{o_i}^{loc}\leq T$. The energy consumed by local computation is given by $E_{o_i}^{loc}=k_{o_i}f_{o_i}^2l_{o_i}^{loc}\phi_{o_i}$, where $k_{o_i}$ represents the computational energy efficiency factor \cite{24}.
\item \emph{Computing offloading:} Let $l_{o_{i}}^{ap}$ denote the amount of SD $o_i$’s data processed at the HAP, and $l_{o_i}^{p_i}$ denote the amount of SD $o_i$’s data processed at the AD $p_i$. Therefore, the time required for SD $o_{i}$ to offload its data is given by $t_{o_{i}}^{off}=\frac{l_{o_{i}}^{ap}+l_{o_{i}}^{p_{i}}}{R_{o_{i}}^{ap}}$, where $R_{o_{i}}^{ap}=B\log_{2}(1+\frac{p_{o_i}h_{o_{i}}}{N_{0}})$ denotes the data transmission rate between SD $o_i$ and the HAP. Here, $B$ denotes the sub-channel bandwidth, $N_0$ denotes the receiver noise power, and $p_{o_i}$ represents the transmit power at SD $o_i$. The offloading time of SD $o_i$ needs to satisfy $t_{o_i}^{off}\leq\alpha_2T$, since the HAP receives the offloaded data during ${\alpha}_{2}T$, as shown in Fig. \ref{timeframe}. Additionally, the energy consumption due to offloading the data is given by $E_{o_i}^{off}=p_{o_i}t_{o_i}^{off}=\frac{p_{o_i}(l_{o_i}^{ap}+l_{o_i}^{p_i})}{R_{o_i}^{ap}}$.
\end{itemize}

According to the energy causality property, an IoT device cannot consume more energy than it harvests within a given time frame. Hence, the energy consumption of SD $o_i$ must adhere to the condition: $ E_{o_i}^{loc}+E_{o_i}^{off}\leq E_{o_i}^h$. Furthermore, the minimum data processing constraint dictates that each IoT device’s data needs to be processed at least $l_{th}$ bits per time frame. Therefore, the total amount of data from SD $o_i$ that needs to be processed must satisfy the following inequality:
\begin{align}
\label{minimum data processing constraint}
l_{o_i}^{loc}+l_{o_i}^{ap}+l_{o_i}^{p_i}\geq l_{th}.
\end{align}

\subsubsection{Auxiliary Devices}

The data of AD $p_i$ is processed exclusively on the local device. However, its energy consumption can be divided into two main components: local processing energy consumption and auxiliary processing energy consumption. This is due to the need to process a portion of SD $o_i$’s data. These two aspects are analyzed as follows.
\begin{itemize}
\item \emph{Auxiliary processing:} The time required for AD $p_i$ to complete auxiliary processing can be calculated as $t_{p_i}^{o_i}=\frac{l_{o_i}^{p_i}\phi_{o_i}}{f_{p_i}}$, where $f_{p_i}$ represents the computational speed of AD $p_i$'s processor. Additionally, auxiliary processing can only start after AD $p_i$ receives the data from the HAP. The transmission time is denoted as $t_{ap}^{p_i}=\frac{l_{o_i}^{p_i}}{R_{ap}^{p_i}}$, where $R_{ap}^{p_i}=Blog_2(1+\frac{Ph_{p_i}}{N_0})$ denotes the data transmission rate between AD $p_i$ and the HAP. Therefore, as shown in Fig. \ref{timeframe}, the following inequality must be satisfied:
\begin{align}
\label{time constraint1 of AD}
(\alpha_1+\alpha_2)T + t_{ap}^{p_i} + t_{p_i}^{o_i} \leq T.
\end{align}
Moreover, the energy consumption resulting from auxiliary processing can be calculated as $E_{p_i}^{o_i}=k_{p_i}f_{p_i}^2l_{o_i}^{p_i}\phi_{o_i}$.

\item \emph{Local processing:} Let $l_{p_i}^{loc}$ denote the amount of AD $p_i$’s data processed locally. The time required for AD $p_i$ to complete local processing can be calculated as $t_{p_i}^{loc}=\frac{l_{p_i}^{loc}\phi_{p_i}}{f_{p_i}}$. Additionally, the minimum data processing constraint requires that the total amount of data from AD $p_i$ that needs to be processed satisfies the following inequality:
\begin{align}
\label{minimum data processing constraint of AD}
l_{p_{i}}^{loc}\geq l_{th}.
\end{align}
And the energy consumption resulting from local processing can be calculated as  $E_{p_{i}}^{loc}=k_{p_{i}}f_{p_{i}}^{3}t_{p_{i}}^{loc}=k_{p_{i}}f_{p_{i}}^{2}l_{p_{i}}^{loc}\phi_{p_i} $. 
\end{itemize}
According to the energy causality property, the energy consumption of AD $p_i$ must adhere to the condition: $E_{p_i}^{loc}+E_{p_i}^{o_i}\leq E_{p_i}^h$. And the total computation time of AD $p_i$ includes both auxiliary processing and local processing, so the following inequality must hold:
\begin{align}
\label{time constraint2 of AD}
t_{p_i}^{loc}+t_{p_i}^{o_i}\leq T.
\end{align}

\subsubsection{Independent Devices}

The IDs do not form collaborative clusters with other devices. Consequently, the data of ID $q_k$ is divided into two parts: one part is processed locally, and the other is processed on the HAP. In terms of energy consumption, the energy expenditure of ID $q_k$ can be divided into two main parts: local processing and computing offloading. Next, these two aspects are analyzed as follows.

\begin{itemize}
\item \emph{Local processing:} Let $l_{q_k}^{loc}$ represent the amount of data allocated for local processing by ID $q_k$. Consequently, the local processing delay can be calculated using the formula $t_{q_{k}}^{loc}=\frac{l_{q_{k}}^{loc}\phi_{q_{k}}}{f_{q_{k}}}$, where $f_{q_k}$ represents the computational speed of ID $q_k$’s processor. The local processing energy consumption is given by $E_{q_{k}}^{loc}=k_{q_{k}}f_{q_{k}}^{2}l_{q_{k}}^{loc}\phi_{q_{k}}$. Similarly, the local processing delay needs to satisfy the inequality: $t_{q_{k}}^{loc}\leq T$.
\item \emph{Computing offloading:} Let $l_{q_{k}}^{ap}$ represent the amount of data from ID $q_k$ that is processed at the HAP. The offloading time can be calculated using the formula $t_{q_k}^{off}=\frac{l_{q_k}^{ap}}{R_{q_k}^{ap}}$, where $ R_{q_{k}}^{ap}=Blog_{2}(1+\frac{p_{q_k}h_{q_{k}}}{N_{0}}) $ represents the data transmission rate between ID $q_k$ and the HAP. Similarly, the offloading time of ID $q_k$ needs to satisfy the condition $t_{q_k}^{off}\leq\alpha_2T$. Additionally, the energy consumption due to data offloading can be obtained as $E_{q_{k}}^{off}=p_{q_k}t_{q_{k}}^{off}$.
\end{itemize}
Similarly, the total energy consumption of ID $q_k$ must satisfy the energy causality constraint represented by $ E_{q_{k}}^{loc}+E_{q_{k}}^{off}\leq E_{q_{k}}^{h} $. Moreover, there exists a constraint on the minimum amount of data to be processed, expressed as $ l_{q_{k}}^{loc}+l_{q_{k}}^{ap}\geq l_{th}$.

\subsubsection{Hybrid Access Point}

The HAP initially transmits RF signals within $\alpha_1 T$ and subsequently receives data offloaded by IoT devices within ${\alpha}_{2}T$. Processing commences only after data reception. Consequently, the processing time for the HAP is $(1-\alpha_1-\alpha_2)T$, and the workload which needs to be processed by the HAP is given by $ \sum_{i=1}^{m}l_{o_i}^{ap}\phi_{o_i}+\sum_{k=1}^{N-2m}l_{q_k}^{ap}\phi_{q_k}$. Thus, the following inequality must be satisfied: 
\begin{align}
\label{minimum data processing constraint of AD}
\frac{\sum_{i=1}^{m}l_{o_i}^{ap}\phi_{o_i}+\sum_{k=1}^{N-2m}l_{q_k}^{ap}\phi_{q_k}}{f_{ap}}\leq(1-\alpha_1-\alpha_2)T,
\end{align}
where $f_{ap}$ represents the processor’s computing speed (cycles per second) of the HAP.

\section{Problem Formulation and Analysis}

Our study aims to maximize the WSCR of the IIoT system by jointly optimizing the collaboration strategy $\boldsymbol{\psi} = \{m,\boldsymbol{o},\boldsymbol{p},\boldsymbol{q}\}$ among IoT devices and the HAP, the system time allocation $\boldsymbol{\alpha} = \{\alpha_{1},\alpha_{2}\}$, and the data allocation $\boldsymbol{l} = \{\boldsymbol{l_o,l_p,l_q}\}$, where $\boldsymbol{l_o}=[l_{o_1}^{loc},l_{o_1}^{ap},l_{o_1}^{p_1},...,l_{o_m}^{loc},l_{o_m}^{ap},l_{o_m}^{p_m}]^T$, $\boldsymbol{l_p}=[l_{p_1}^{loc},...,l_{p_m}^{loc}]^T$, and $\boldsymbol{l_q}=[l_{q_1}^{loc},l_{q_1}^{ap},...,l_{q_{N-2m}}^{loc},l_{q_{N-2m}}^{ap}]^T$. Additionally, the optimization is conducted considering energy causality, the minimum data processing requirements, and time constraints. Mathematically, the WSCR can be formulated as
\begin{align}
\mathcal{P}_1: \quad \max_{\boldsymbol{\boldsymbol{\psi,\alpha,l}}} & \sum_{i=1}^m\{w_{o_i}(l_{o_i}^{loc}+l_{o_i}^{ap}+l_{o_i}^{p_i})+ w_{p_i}l_{p_i}^{loc}\} \notag \\
& + \sum_{k=1}^{N-2m}w_{q_k}(l_{q_k}^{loc}+l_{q_k}^{ap}) \\
\text{s.t.} \quad (7a):&|\boldsymbol{o}|=|\boldsymbol{p}|=m\in \{0,1,2,...\left\lfloor\frac{N}{2}\right\rfloor \} \notag \\
(7b):&\boldsymbol{o}\cup\boldsymbol{p}\cup\boldsymbol{q}=\mathcal{N} \notag \\
(7c):&\frac{l_{o_i}^{loc}\phi_{o_i}}{f_{o_i}}\le T, \forall i \in \{1,2,...m\} \notag \\
(7d):&\frac{l_{o_i}^{ap}+l_{o_i}^{p_i}}{R_{o_i}^{ap}}\le \alpha_2T, \forall i \in \{1,2,...m\} \notag \\
(7e):&\frac{\phi_{o_i}l_{o_i}^{p_i}}{f_{p_i}}\le (1-\alpha_1-\alpha_2)T-\frac{l_{o_i}^{p_i}}{R_{ap}^{p_i}}, \forall i \notag \\
(7f):& \frac{\phi_{p_i}l_{p_i}^{loc}}{f_{p_i}}\le T-\frac{\phi_{o_i}l_{o_i}^{p_i}}{f_{p_i}}, \forall i \in \{1,2,...,m\} \notag \\
(7g):&\frac{l_{q_k}^{loc}\phi_{q_k}}{f_{q_k}} \le T,\forall k \in \{1,2,...,N-2m\} \notag \\
(7h):&\frac{l_{q_k}^{ap}}{R_{q_k}^{ap}}\le \alpha_2T,\forall k\in \{1,2,...,N-2m\} \notag \\
(7i):& \frac{\sum_{i=1}^m \phi_{o_i}l_{o_i}^{ap}+\sum_{k=1}^{N-2m}\phi_{q_k}l_{q_k}^{ap}}{f_{ap}}\le (1-\alpha_1-\alpha_2)T \notag \\
(7j):&\frac{l_{o_i}^{p_i}}{R_{ap}^{p_i}}\le(1-\alpha_1-\alpha_2)T,\forall i \in \{1,2,...,m\} \notag \\
(7k):& \alpha_1\in (0,1],\alpha_2 \in[0,1),\alpha_1+\alpha_2\le1 \notag \\
(7l):&k_{o_i}f_{o_i}^2l_{o_i}^{loc}\phi_{o_i}+\frac{p_{o_i}(l_{o_i}^{ap}+l_{o_i}^{p_i})}{R_{o_i}^{ap}}\le E_{o_i}^h,\forall i \notag \\
(7m):&k_{p_i}f_{p_i}^2l_{p_i}^{loc}\phi_{p_i}+k_{p_i}f_{p_i}^2l_{o_i}^{p_i}\phi_{o_i}\le E_{p_i}^h ,\forall i\notag\\
(7n):&k_{q_k}f_{q_k}^2l_{q_k}^{loc}\phi_{q_k}+\frac{p_{q_k}l_{q_k}^{ap}}{R_{q_k}^{ap}} \le E_{q_k}^h, \forall k \notag\\
(7o):&l_{o_i}^{loc} + l_{o_i}^{ap}+l_{o_i}^{p_i} \ge l_{th},\forall i \in \{1,2,...,m\} \notag\\
(7p):&l_{p_i}^{loc} \ge l_{th},\forall i \in \{1,2,...,m\} \notag\\
(7q):&l_{q_k}^{loc}+l_{q_k}^{ap} \ge l_{th},\forall k \in \{1,2,...,N-2m\}\notag\\
(7r):&l_{o_i}^{loc},l_{o_i}^{ap},l_{o_i}^{p_i},l_{q_k}^{loc},l_{q_k}^{ap}\ge0,\forall i,\forall k,\notag
\end{align}
where $w_{o_i},w_{p_i},w_{q_k}$ are unique weights assigned to SD $o_i$, AD $p_i$, and ID $q_k$, respectively. Constraints $(7a)$ and $(7b)$ pertain to the formation of collaborative clusters, requiring an SD to collaborate with a unique AD. Constraints $(7c)$ and $(7d)$ concern the time overhead of SDs. Constraints $(7e)$ and $(7f)$ relate to the time overhead of ADs. Constraints $(7g)$ and $(7h)$ involve the time overhead of IDs. Constraints $(7i)$ to $(7k)$ govern the time allocation of the HAP. Constraints $(7l)$ to $(7n)$ pertain to the energy consumption causality for SDs, ADs, and IDs, respectively. Constraints $(7o)$ to $(7q)$ specify the minimum data processing requirements for SDs, ADs, and IDs, respectively. Constraint $(7r)$ pertains to the non-negativity for the data allocation.

According to our assumption, the energy consumed by the $n$th device for continuous data processing within a time frame should be greater than the energy collected\footnote{In wireless-powered MEC systems, wireless devices (WDs) are typically energy-constrained. This means a WD can exhaust all harvested energy within a given time frame by operating at maximum computing speed. Thus, in accordance with \cite{17,20,25,31}, we adopt this assumption. Exploring the potential remaining energy of each IoT device after a time frame will be a future focus of our research.}. Mathematically, this can be expressed as $k_nf_n^3T> E_n^h$. Based on this assumption, we give the following lemma.
\begin{lemma}
At the end of each time frame, all devices exhaust the energy collected during that time frame.
\end{lemma}
\begin{proof}
Assume that under the optimal solution of problem $\mathcal{P}_1$, there is a device that has not exhausted its energy. Based on the above assumption, we can increase the amount of data processed locally by that device to exhaust the remaining energy without violating the constraints in $\mathcal{P}_1$. As a result, the WSCR is clearly enhanced, which contradicts the initial assumptions. This proves Lemma 1. Hence, we can disregard the remaining energy from the last time frame.
\end{proof}

Several lemmas are given next to simplify the constraints of the problem.

\begin{lemma}
Constraints $(7c)$ and $(7g)$ must hold.
\end{lemma}

\begin{proof}Given $k_nf_n^3T>E_n^h$, it follows that the actual local processing time of SD $o_i$ cannot exceed $T$, i.e., $ \frac{l_{o_{i}}^{loc}\phi_{o_i}}{f_{o_{i}}}<T$, thus satisfying constraint $(7c)$. Similarly, it can be demonstrated that constraint $(7g)$ must also be met.\end{proof}

\begin{lemma}
The constraint $(7j)$ must hold.
\end{lemma}
\begin{proof}
Due to constraint $(7e)$, $\frac{l_{o_i}^{p_i}}{R_{qp}^{p_i}}\leq(1-\alpha_1-\alpha_2)T-\frac{\phi_{o_i}l_{o_i}^{p_i}}{f_{p_i}}$. Additionally, based on constraint $(7r)$, it can be derived that $\frac{l_{o_i}^{p_i}}{R_{ap}^{p_i}}\le(1-\alpha_1-\alpha_2)T$. As a result, constraint $(7j)$ can be removed.
\end{proof}

\begin{lemma}
The constraint $(7k)$ must hold.
\end{lemma}
\begin{proof}
From constraints $(7i)$ and $(7r)$, it follows that $\begin{aligned}0\le&\frac{\sum_{i=1}^m\phi_{o_i}l_{o_i}^{ap}+\sum_{k=1}^{N-2m}\phi_{q_k}l_{q_k}^{ap}}{f_{ap}}\le(1-\alpha_1-\alpha_2)T\end{aligned}$, which implies $ \alpha_1+\alpha_2\le 1$. Moreover, from constraints $(7h)$, $(7l)$ and $(7o)$, it is evident that $ \alpha_1>0$ and $ \alpha_2\geq0$. Consequently, we can infer that $ \alpha_1\in(0,1]$ and $ \alpha_2\in[0,1)$. This establishes that constraint $(7k)$ holds.
\end{proof}

Additionally, according to (\ref{harvested energy}), it is apparent that constraints $(7l)-(7n)$ display nonlinearity. To simplify problem resolution, we reformulate these three constraints as:
\begin{align}
\label{constraints}
&k_{o_i}f_{o_i}^2l_{o_i}^{loc}\phi_{o_i}+\frac{p_{o_i}(l_{o_i}^{ap}+l_{o_i}^{p_i})}{R_{o_i}^{ap}}\le \eta Ph_{o_i}\alpha_1T,\forall i, \notag \\
&k_{o_i}f_{o_i}^2l_{o_i}^{loc}\phi_{o_i}+\frac{p_{o_i}(l_{o_i}^{ap}+l_{o_i}^{p_i})}{R_{o_i}^{ap}}\le E_{o_i}^{max},\forall i, \notag \\
&k_{p_i}f_{p_i}^2l_{p_i}^{loc}\phi_{p_i}+k_{p_i}f_{p_i}^2l_{o_i}^{p_i}\phi_{o_i}\le \eta Ph_{p_i}\alpha_1T,\forall i, \\
&k_{p_i}f_{p_i}^2l_{p_i}^{loc}\phi_{p_i}+k_{p_i}f_{p_i}^2l_{o_i}^{p_i}\phi_{o_i}\le E_{p_i}^{max},\forall i, \notag \\
&k_{q_k}f_{q_k}^2l_{q_k}^{loc}\phi_{q_k}+\frac{p_{q_k}l_{q_k}^{ap}}{R_{q_k}^{ap}} \le \eta Ph_{q_k}\alpha_1T,\forall k, \notag \\
&k_{q_k}f_{q_k}^2l_{q_k}^{loc}\phi_{q_k}+\frac{p_{q_k}l_{q_k}^{ap}}{R_{q_k}^{ap}} \le E_{q_k}^{max},\forall k. \notag
\end{align}

Building upon the lemmas and constraint transformations mentioned earlier, the original problem can be reduced to:
\begin{align}
\label{P2}
\mathcal{P}_2: \quad \max_{\boldsymbol{\boldsymbol{\psi,\alpha,l}}} & \sum_{i=1}^m\{w_{o_i}(l_{o_i}^{loc}+l_{o_i}^{ap}+l_{o_i}^{p_i})+ w_{p_i}l_{p_i}^{loc}\}  \\
& + \sum_{k=1}^{N-2m}w_{q_k}(l_{q_k}^{loc}+l_{q_k}^{ap}) \notag \\
\text{s.t.} \;
(9a):&(7a),(7b),(7d)-(7f),(7h),(7i),(7o)-(7r),(8)\notag
\end{align}

\vspace{-0.1cm}
\begin{lemma}
The problem $\mathcal{P}_2$ is NP-hard.
\end{lemma}
\begin{proof}
To establish the NP-hardness of $\mathcal{P}_2$, we present a specific case of the problem and demonstrate its NP-hardness. Assume that given the number of collaborative clusters $m>0$, the construction of collaborative clusters is specified as $ \{(i,i+m)\mid i\in\{1,2,... .m\}\}$, where ${(i,i+m)}$ represents the indices of the corresponding IoT devices. Consequently, we can obtain $ \boldsymbol{o}=\{i\mid i\in\{1,2,...m\}\}, \boldsymbol{p}=\{i+m\mid i\in\{1,2,...m\}\}$, and $ \boldsymbol{q}=\{2m+1,2m+2,....N\}$. Based on this, problem $\mathcal{P}_2$ is reduced to $\mathcal{P}_3$. It can be observed that the objective function and constraints of $\mathcal{P}_3$ are both linear. Additionally, the optimization variable $m$ is an integer. As a result, $\mathcal{P}_3$ constitutes a mixed-integer linear programming (MILP) problem known as NP-hard. Therefore, this shows that problem $\mathcal{P}_2$ is also NP-hard.
\end{proof}
\begin{align}
\label{P3}
\mathcal{P}_3: \quad \max_{m,\boldsymbol{\alpha,l}}  &\sum_{i=1}^m\{w_{o_i}(l_{o_i}^{loc}+l_{o_i}^{ap}+l_{o_i}^{p_i})+ w_{p_i}l_{p_i}^{loc}\} \\
& + \sum_{k=1}^{N-2m}w_{q_k}(l_{q_k}^{loc}+l_{q_k}^{ap}) \notag \\
\text{s.t.} \quad
(10a):&(7d)-(7f),(7h),(7i),(7o)-(7r),(8)\notag\\
(10b):&m\in\{1,2,...\left\lfloor\frac{N}{2}\right\rfloor\}\notag
\end{align}

\begin{figure}[t]
\centerline{\includegraphics[width=0.45\textwidth]{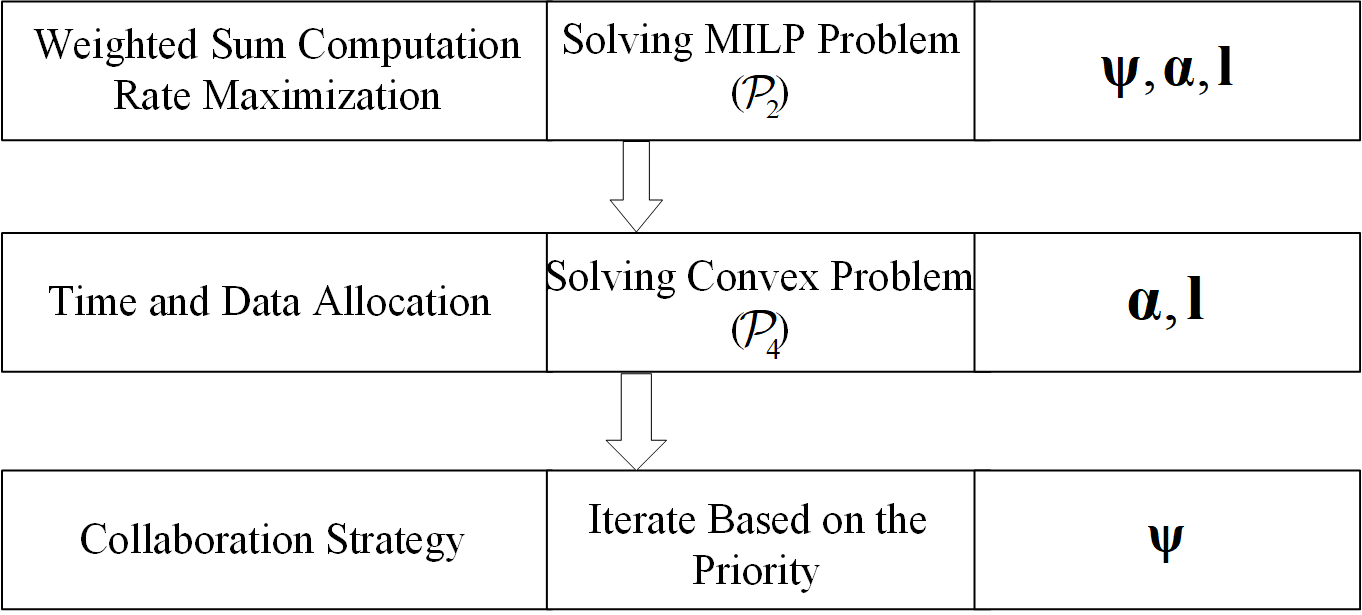}}
\caption{The two-level optimization structure for sloving $\mathcal{P}_2$.}
\label{solution}

\end{figure}

\section{Solution Approach}
In this section, efficient algorithms are proposed to solve problem $\mathcal{P}_2$. As shown in Fig. \ref{solution}, problem ($\mathcal{P}_2$) can be decomposed into two sub-problems, namely, time and data allocation ($\mathcal{P}_4$) and collaboration strategy. Thus, we can optimize the time and data allocation $\{\boldsymbol{\alpha}, \boldsymbol{l}\}$ for a given collaboration strategy $\boldsymbol{\psi}^0$ first. Then, we can search among all possible collaboration strategies to find an optimal or satisfying sub-optimal collaboration strategy $\boldsymbol{\psi}^*$. However, due to the exponentially large search space, it is computationally prohibitive when $N$ is large. Therefore, a low-complexity algorithm is needed. Next, we provide detailed solutions for these two subproblems.

\subsection{Time and Data Allocation}\label{IV-A}
When the collaborative strategy $\boldsymbol{\psi}$ is fixed, $\mathcal{P}_2$ transforms into the following problem:

\begin{align}
\label{P4}
\mathcal{P}_4: \quad \max_{\boldsymbol{\boldsymbol{\alpha,l}}} & \sum_{i=1}^m\{w_{o_i}(l_{o_i}^{loc}+l_{o_i}^{ap}+l_{o_i}^{p_i})+ w_{p_i}l_{p_i}^{loc}\}  \\
& + \sum_{k=1}^{N-2m}w_{q_k}(l_{q_k}^{loc}+l_{q_k}^{ap}) \notag \\
\text{s.t.} \;
(11a):&(7d)-(7f),(7h),(7i),(7o)-(7r),(8)\notag
\end{align}

Problem $\mathcal{P}_4$ is a convex optimization problem because both the objective function and constraints are linear in terms of the optimization variables. In order to facilitate the solution of the problem, $\mathcal{P}_4$ is first transformed into the standard matrix form:
\begin{align}
\label{P5}
\mathcal{P}_5: \quad \min_{\boldsymbol{\boldsymbol{x}_1}}\quad & \boldsymbol{c}_1^T\boldsymbol{x}_1 \\
\text{s.t.} \quad
(12a):&\boldsymbol{A}_1\boldsymbol{x}_1\preceq\boldsymbol{b}\notag\\
(12b):&\boldsymbol{x}_1\succeq \boldsymbol{0}.\notag
\end{align}

The vector $\boldsymbol{x}_1$ represents the optimization variables, and it is defined as $\boldsymbol{x}_1=[\alpha_1,\alpha_2,l_{o_1}^{loc},l_{o_1}^{ap},l_{o_1}^{p_1},l_{p_1}^{loc},...,l_{o_m}^{loc},l_{o_m}^{ap},l_{o_m}^{p_m},l_{p_m}^{loc},l_{q_1}^{loc},l_{q_1}^{ap},...,\\l_{q_{N-2m}}^{loc},l_{q_{N-2m}}^{ap}]^T$. The weight coefficients corresponding to the optimization variables in the objective function of $\mathcal{P}_4$ are represented by $\boldsymbol{c}_1$, given as $\boldsymbol{c}_1=[0,0,-w_{o_1},-w_{o_1},-w_{o_1},-w_{p_1},...,-w_{o_m},-w_{o_m},-w_{o_m},-\\w_{p_m},-w_{q_1},-w_{q_1},...,-w_{q_{N-2m}},-w_{q_{N-2m}}]^T$. The matrix $\boldsymbol{A}_{1}$ represents the coefficients in the constraints, which has a total of $2N+2$ columns and $4N+m+1$ rows. The vector $\boldsymbol{b}$ represents the constants on the right-hand side of the inequality in the constraints, and it has a total of $4N+m+1$ elements.

Introducing slack variables allows for the transformation of problem $\mathcal{P}_5$ into the following form:
\begin{align}
\label{P6}
\mathcal{P}_6: \quad \min_{\boldsymbol{x}_1,\boldsymbol{x}_2}\quad & \boldsymbol{c}_1^T\boldsymbol{x}_1+\boldsymbol{c}_2^T\boldsymbol{x}_2 \\
\text{s.t.} \quad
(13a):&\boldsymbol{A}_{1}\boldsymbol{x}_{1}+\boldsymbol{x}_{2}=\boldsymbol{b},\notag\\
(13b):&\boldsymbol{x_1},\boldsymbol{x_2}\succeq\boldsymbol{0},\notag
\end{align}

$\boldsymbol{x}_2$ is a vector of slack variables with $4N+m+1$ elements, while $\boldsymbol{c}_2=\boldsymbol{0}$ and its length is also $4N+m+1$.

Let $ \boldsymbol{c}_{3}=[\boldsymbol{c}_{1};\boldsymbol{c}_{2}],\boldsymbol{x}_{3}=[\boldsymbol{x}_{1};\boldsymbol{x}_{2}],\boldsymbol{A}_{2}=[\boldsymbol{A}_{1},\boldsymbol{I}]$, the problem $\mathcal{P}_6$ can be equated to:
\begin{align}
\label{P7}
\mathcal{P}_7: \quad \min_{\boldsymbol{\boldsymbol{x}_3}}\quad & \boldsymbol{c}_3^T\boldsymbol{x}_3 \\
\text{s.t.} \quad
(14a):&\boldsymbol{A}_2\boldsymbol{x}_3=\boldsymbol{b}\notag\\
(14b):&\boldsymbol{x}_3\succeq \boldsymbol{0}.\notag
\end{align}

To eliminate the inequality constraints in the problem $\mathcal{P}_7$, a barrier function is introduced for transformation:
\begin{align}
\label{P8}
\mathcal{P}_8: \quad \min_{\boldsymbol{\boldsymbol{x}_3}}\quad & \boldsymbol{c}_3^T\boldsymbol{x}_3-\mu\sum_{i=1}^{6N+m+3}log(x_3^i) \\
\text{s.t.} \quad
(15a):&\boldsymbol{A}_2\boldsymbol{x}_3=\boldsymbol{b},\notag
\end{align}
where $x_3^i$ represents the $i$th element in $\boldsymbol{x}_3$ and $\mu>0$. Additionally, the solution obtained by solving problem $\mathcal{P}_8$ gets closer to the optimal solution of problem $\mathcal{P}_7$ as the value of $\mu$ decreases.

It can be observed that problem $\mathcal{P}_8$ remains a convex optimization problem, allowing the utilization of the Lagrangian duality method to obtain an optimal solution when the value of $\mu$ is given. The Lagrangian function for $\mathcal{P}_8$ is given by:
\begin{align}
\label{e9}
L_{\mu}\left(\boldsymbol{x}_3,\boldsymbol{v}\right)=\boldsymbol{c}_3^T\boldsymbol{x}_3-\mu\sum_{i=1}^{6N+m+3}\log(x_3^i)\\+\boldsymbol{v}^T\left(\boldsymbol{A}_2\boldsymbol{x}_3-\boldsymbol{b}\right)\notag,
\end{align}
where vector $\boldsymbol{v}$ comprises Lagrange multipliers. By applying the Karush-Kuhn-Tucker (KKT) conditions and setting $\lambda_i=\frac\mu{x_3^i}$, $\boldsymbol{\lambda}=[\lambda_1,...\lambda_{6N +m+3}]^T$, $\boldsymbol{X}=diag(\boldsymbol{x}_3)$, $ \boldsymbol{\Lambda}=diag(\boldsymbol{\lambda})$, and $ \boldsymbol{1}=\{1,...,1\}^T$, we can obtain the following results:
\begin{align}
\label{e10}
\boldsymbol{A}_2\boldsymbol{x}_3=\boldsymbol{b},
\end{align}
\begin{align}
\label{e11}
\boldsymbol{X}\boldsymbol{\Lambda}\boldsymbol{1}=\mu\boldsymbol{1},
\end{align}
and
\begin{align}
\label{e12}
\boldsymbol{A}_2^T\boldsymbol{v}+\boldsymbol{c}_3=\boldsymbol{\lambda}.
\end{align}

Based on the above transformations, the solution to problem $\mathcal{P}_8$ can be determined by solving the following equation:
\begin{align}
\label{e13}
F_\mu(\boldsymbol{\lambda},\boldsymbol{v},\boldsymbol{x}_3)=\begin{pmatrix}\boldsymbol{A}_2\boldsymbol{x}_3-\boldsymbol{b}\\\boldsymbol{X}\boldsymbol{\Lambda}\boldsymbol{1}-\mu\boldsymbol{1}\\\boldsymbol{A}_2^T\boldsymbol{v}+\boldsymbol{c}_3-\boldsymbol{\lambda}\end{pmatrix}=\boldsymbol{0}.
\end{align}

The Newton iterative method can be employed to solve the problem. Specifically, in the $(k+1)$th iteration, the first-order Taylor expansion of function $F_\mu(\boldsymbol{\lambda},\boldsymbol{v},\boldsymbol{x}_3)$ is given by:

\begin{align}
\label{e14}
F_{\mu}(\boldsymbol{\lambda},\boldsymbol{v},\boldsymbol{x}_{3})\approx &F_{\mu}(\boldsymbol{\lambda}^{k},\boldsymbol{v}^{k},\boldsymbol{x}_{3}^{k})\notag\\+&\nabla F_{\mu}(\boldsymbol{\lambda}^{k},\boldsymbol{v}^{k},\boldsymbol{x}_{3}^{k})\left(\begin{array}{l}{\boldsymbol{\lambda}-\boldsymbol{\lambda}^{k}}\\{\boldsymbol{v}-\boldsymbol{v}^{k}}\\{\boldsymbol{x}_{3}-\boldsymbol{x}_{3}^{k}}\\\end{array}\right).
\end{align}
Here, $\boldsymbol{\lambda}^k$, ${\boldsymbol{v}}^k$, and ${\boldsymbol{x}}_3^k$ represent the solutions obtained in the $k$th iteration. Therefore, solving the equation $F_{\mu}(\boldsymbol{\lambda},\boldsymbol{v},\boldsymbol{x}_{3})=\boldsymbol{0}$ is equivalent to solving the following equation:

\begin{align}
\label{e15}
F_\mu(\boldsymbol{\lambda}^k,\boldsymbol{v}^k,\boldsymbol{x}_3^k)+\nabla F_\mu(\boldsymbol{\lambda}^k,\boldsymbol{v}^k,\boldsymbol{x}_3^k)\left(\begin{array}{c}\boldsymbol{\lambda}-\boldsymbol{\lambda}^k\\\boldsymbol{v}-\boldsymbol{v}^k\\\boldsymbol{x}_3-\boldsymbol{x}_3^k\end{array}\right)=\boldsymbol{0},
\end{align}
where
\begin{align}
\label{Jac}
\nabla F_\mu(\boldsymbol{\lambda},\boldsymbol{v},\boldsymbol{x}_3)=\begin{pmatrix}\boldsymbol{0}&\boldsymbol{0}&\mathbf{A}_2\\\boldsymbol{X}&\boldsymbol{0}&\boldsymbol{\Lambda}\\-\mathbf{I}&\mathbf{A}_2^T&\boldsymbol{0}\end{pmatrix}.
\end{align}

Hence, determining the optimal solution involves solving the following equations:
\begin{align}
\label{equ}
\begin{pmatrix}\boldsymbol{0}&\boldsymbol{0}&\boldsymbol{A}_2\\\boldsymbol{X}^k&\boldsymbol{0}&\boldsymbol{\Lambda}^k\\\boldsymbol{-I}&\boldsymbol{A}_2^T&\boldsymbol{0}\end{pmatrix}\begin{pmatrix}\Delta\boldsymbol{\lambda}\\\Delta\boldsymbol{v}\\\Delta\boldsymbol{x}_3\end{pmatrix}=\begin{pmatrix}-\boldsymbol{\gamma}_1^k\\-\boldsymbol{\gamma}_2^k\\-\boldsymbol{\gamma}_3^K\end{pmatrix},
\end{align}
where
\begin{align}
\left(\begin{array}{c}\boldsymbol{\gamma}_{1}^{k}\\\boldsymbol{\gamma}_{2}^{k}\\\boldsymbol{\gamma}_{3}^{k}\end{array}\right)=F_{\mu}(\boldsymbol{\lambda}^{k},\boldsymbol{v}^{k},\boldsymbol{x}_{3}^{k}),
\label{eq25}
\end{align}
and
\begin{align}
\label{}
\begin{pmatrix}\Delta\boldsymbol{\lambda}\\\Delta\boldsymbol{v}\\\Delta\boldsymbol{x}_3\end{pmatrix} = \left(\begin{array}{c}\boldsymbol{\lambda}-\boldsymbol{\lambda}^k\\\boldsymbol{v}-\boldsymbol{v}^k\\\boldsymbol{x}_3-\boldsymbol{x}_3^k\end{array}\right).
\end{align}
The solution obtained by the $k+1$th iteration process is:
\begin{align}
\label{update}
\left(\begin{array}{c}\boldsymbol{\lambda}^{k+1}\\\boldsymbol{v}^{k+1}\\\boldsymbol{x}^{k+1}\end{array}\right)=\left(\begin{array}{c}\boldsymbol{\lambda}^{k}\\\boldsymbol{v}^{k}\\\boldsymbol{x}^{k}\end{array}\right)+{\beta}\left(\begin{array}{c}\Delta\boldsymbol{\lambda}\\\Delta\boldsymbol{v}\\\Delta\boldsymbol{x}_{3}\end{array}\right),
\end{align}
where $ \beta$ is the iteration step. Thus, we can derive an approximate optimal solution to the problem $\mathcal{P}_8$ for a given value of $\mu$. Next, we can decrease the value of $\mu$ and utilize the solution obtained with the previous $\mu$ value as the initial solution for the problem $\mathcal{P}_8$ with the updated $\mu$ value. Afterwards, we can apply the same approach to solve the updated problem $\mathcal{P}_8$. The iterative process continues until the accuracy requirement is satisfied or the maximum number of iterations is reached. This approach requires fewer iterations than solving problem $\mathcal{P}_8$ with the minimum $\mu$ value, accelerating the solution process for problem $\mathcal{P}_5$. In order to further accelerate the solution of problem $\mathcal{P}_5$, we can simultaneously decrease the value of $\mu$ while iterating to solve problem $\mathcal{P}_8$. This approach is outlined in Algorithm 1.

\begin{figure}[t]
\centerline{\includegraphics[width=0.5\textwidth]{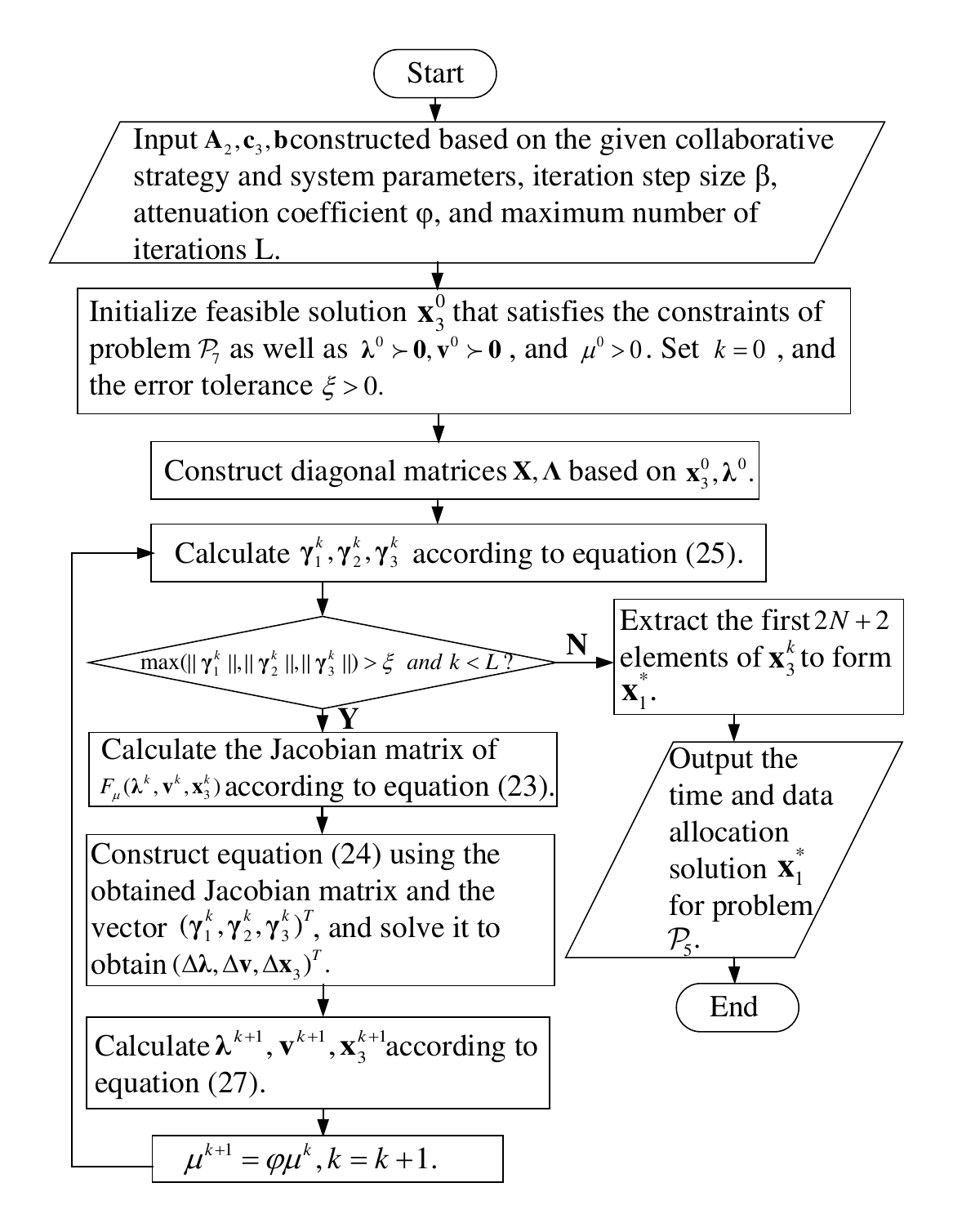}}
\caption{Algorithm 1: Interior-Point Method Based Algorithm for Solving Problem $\mathcal{P}_5$.}

\label{algo1}
\end{figure}

\textit{Remark 1 (Details on the Iteration Process of Algorithm 1): In our designed algorithm, the key factors influencing the iterative process include the convergence judgment criterion, the iteration step size $\beta$, and the attenuation coefficient $\varphi$. The convergence criterion determines if the algorithm has achieved a solution satisfying the accuracy requirement. In this paper, we define it as the largest component of the absolute error vector resulting from substituting the current solution into equation (\ref{eq25}) and comparing it to the target result $\boldsymbol{0}$. To limit excessive iterations, we establish a maximum iteration count $L$. The iteration step size $\beta$ regulates the variable updates' magnitude in each iteration. Typically, the step size choice involves a line search technique, such as the Goldstein Condition \cite{41}. Moreover, the attenuation coefficient $\varphi$ controls the decay of the barrier function in each iteration and is typically chosen within the range of (0, 1). Empirically, common choices for the attenuation coefficient include values such as 0.1, 0.2, or smaller. Experimenting with various values and observing their effects on the algorithm's performance is advisable.} 

\textit{Remark 2 (Computational Complexity Analysis of Algorithm 1): In each iteration, the computational complexity of computing the Jacobian matrix for $F_\mu(\boldsymbol{\lambda}^k,\boldsymbol{v}^k,\boldsymbol{x}_3^k)$ is $\mathcal{O}(N^{3})$. Moreover, solving equation (\ref{equ}) requires $\mathcal{O}(N^{3})$, while calculating the updated $\boldsymbol{\gamma}_1^k,\boldsymbol{\gamma}_2^k,\boldsymbol{\gamma}_3^k$ and its paradigm needs $\mathcal{O}(N^{2})$. Therefore, the computational complexity of Algorithm 1 is approximately $\mathcal{O}(LN^{3})$.} 

The solution for problem $\mathcal{P}_5$ is identical to that for problem $\mathcal{P}_4$. Therefore, given the collaborative strategy, we can derive the optimal solution for time and data allocation. Subsequently, our objective is to identify the most favorable collaborative strategy. To achieve this, we can employ an exhaustive search method to enumerate all feasible solutions $\{\boldsymbol{\psi}_k|k\in\mathcal{K}\}$ for collaborative strategies satisfying constraints $(7a)$ and $(7b)$. Next, we can solve problem $\mathcal{P}_4$ for each of these collaborative strategies. The optimal solution $\boldsymbol{\psi}^*$ can be determined by employing the following equation:

\begin{align}
\label{e19} \boldsymbol{\psi}^*=\arg\max_{\boldsymbol{\psi_k}}f(\boldsymbol{\psi_k}),
\end{align}
where $f(\boldsymbol{\psi_k})=\mathop {\emph{max} }\limits_{{\boldsymbol{\alpha, l}}}\sum_{i=1}^{m_k}\{w_{o_{k,i}}(l_{o_{k,i}}^{loc}+l_{o_{k,i}}^{ap}+l_{o_{k,i}}^{p_{k,i}})+w_{p_{k,i}}l_{p_{k,i}}^{loc}\}+\sum_{j=1}^{N-2m_k}w_{q_{k,j}}(l_{q_{k,j}}^{loc}+l_{q_{k,j}}^{ap})$.

\textit{Remark 3 (Computational Complexity Analysis of the exhaustive search scheme):
Initially, we need to determine the number of collaborative clusters  $m$, with a total of $\bigg\lfloor{\frac{N}{2}}\bigg\rfloor + 1$ possible scenarios. Furthermore, for a given number of collaborative clusters $m>0$, there are $C_N^{2m}(2m-1)!!$ possible scenarios for constructing collaborative clusters. In this case, $C_N^{2m}$ represents the number of scenarios in which $2m$ devices are selected from a total of $N$ IoT devices to form collaborative clusters, and $(2m-1)!!$ denotes the number of scenarios in which collaborative clusters are constructed using $2m$ IoT devices. In summary, there are a total of $1 + \sum_{m=1}^{m=\left\lfloor\frac N2\right\rfloor}C_N^{2m}(2m-1)!!$ possible scenarios for constructing collaborative clusters. For each instance of collaborative cluster construction, the maximum achievable objective value needs to be determined based on Algorithm 1. Thus, the overall computational complexity is $ \mathcal{O}(LN^3(1+\sum_{m=1}^{m=\left\lfloor\frac{N}{2}\right\rfloor}C_N^{2m}(2m-1)!!))$.}

It is evident that the exhaustive search method becomes computationally prohibitive when $N$ is large. Therefore, we must explore a low-complexity algorithm to identify the most favorable collaborative strategy.

\subsection{Collaboration Strategy}\label{IV-B}
This subsection presents a low-complexity, priority-based algorithm designed to find a suboptimal solution for collaboration strategy, offering a trade-off between computational complexity and system performance.

In this study, the diversity among IoT devices is considered in three main aspects: (1) the importance weight $w_n$ assigned to the $n$th device's data, (2) the channel quality $h_n$ between the $n$th device and the HAP, and (3) the local computational capacity $f_n$ of the $n$th device. Based on these three variables, a priority function can be constructed to guide the formation of collaborative clusters.
\begin{align}
\label{Pro}
O_n=w_n\frac{k_nf_n^2\phi_n}{\frac{p_n}{R_n^{ap}}}=\frac{w_nBlog_{2}(1+\frac{p_{n}h_{n}}{N_{0}})k_nf_n^2\phi_n}{p_n}.
\end{align}

The term $k_nf_n^2\phi_n$ represents the energy consumption per unit of data volume processed by the $n$th device, and $\frac{p_n}{R_n^{ap}}$ represents the energy consumption per unit of data volume offloaded by the $n$th device. Therefore, a higher value of $O_n$ indicates a higher priority for the $n$th device to serve as SD. During the construction of a new collaborative cluster, the device with the highest priority among the current IDs is selected as SD, and the device with the lowest priority among the current IDs is selected as AD. Therefore, an iterative algorithm can be employed to determine the optimal number $m^*$ of collaborative clusters. Initially, setting $m$ to 0 indicates that all devices are IDs, and then Algorithm 1 is used to calculate the achievable objective value based on the current collaboration strategy. In each iteration, $m$ is incremented by 1, and a new collaborative cluster is constructed by selecting the device with the highest priority among the current IDs as SD and the device with the lowest priority as AD. Algorithm 1 is used again to calculate the achievable objective value based on the current collaborative strategy. After iterating $\left\lfloor\frac{N}{2}\right\rfloor+1$ times, we can obtain the achievable objective value for each value of $m$. Hence, the final output would be the collaboration strategy corresponding to the maximum achievable objective value. The summarized solution is outlined in Algorithm 2.

\begin{figure}[t]
\centerline{\includegraphics[width=0.5\textwidth]{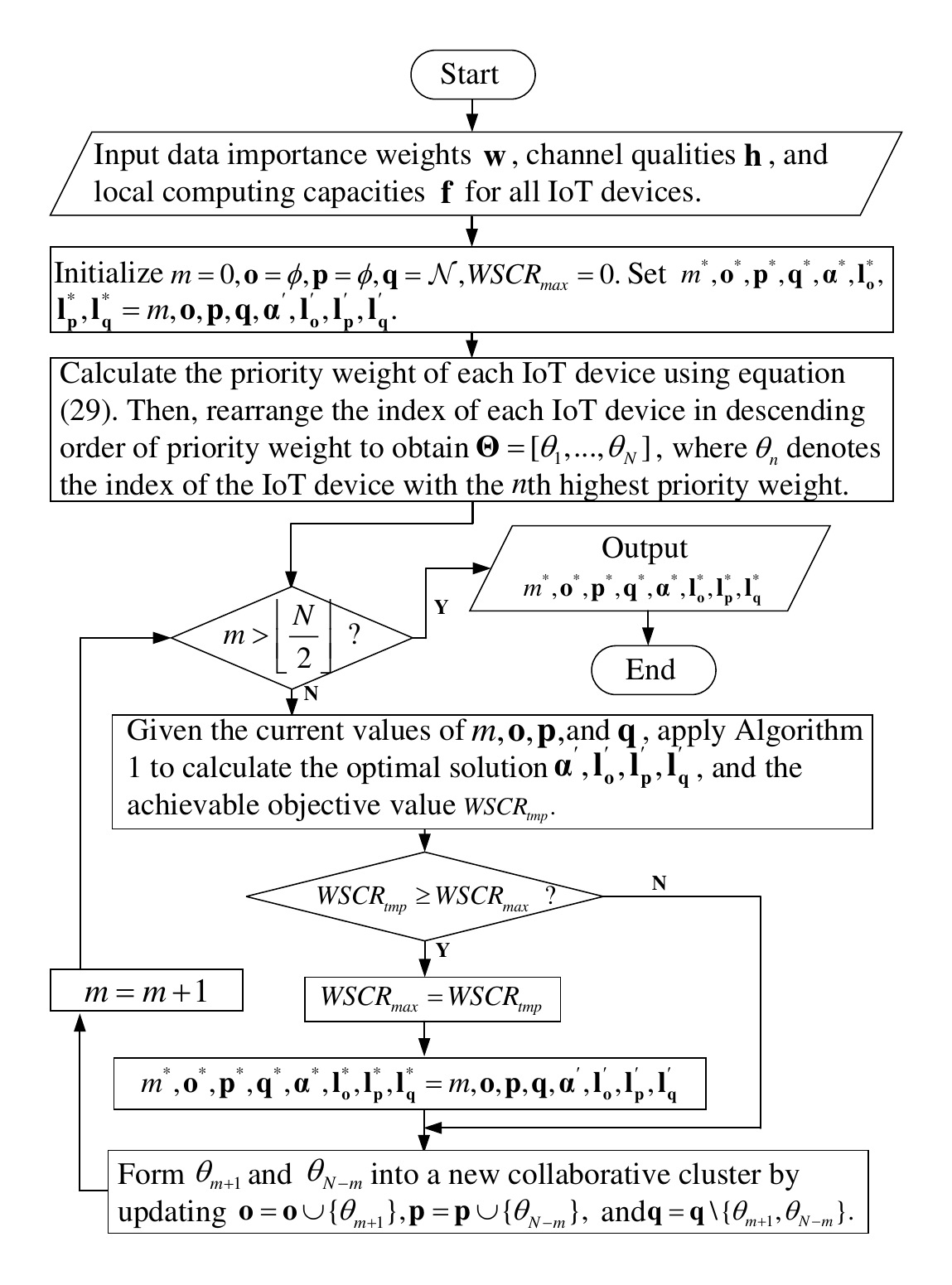}}
\caption{Algorithm 2: Priority-Based Algorithm for Solving Problem $\mathcal{P}_2$.}
\label{algo2}
\end{figure}

\textit{Remark 4 (Computational Complexity Analysis of the priority-based iterative algorithm): The proposed algorithm establishes a priority function to narrow down the search space to $\left\lfloor\frac{N}{2}\right\rfloor$. Additionally, in each iteration, Algorithm 1 is employed, incurring a computational cost of $\mathcal{O}(LN^{3})$. As a result, the overall computational complexity amounts to $\mathcal{O}(LN^4)$.}

Nevertheless, the iterative search process for determining the optimal number of collaborative clusters continues to result in extended execution latency, mainly attributable to the computational complexity of Algorithm 1. Section V compares the algorithm's execution time, illustrating that the iterative search process for the optimal number of collaborative clusters substantially impacts the total execution latency. Thus, we aim to ascertain the optimal number of collaborative clusters with diminished computational complexity. To accomplish this, we put forth a scheme based on deep learning (DL) that encompasses the following primary steps:

\begin{enumerate}
\item{\emph{Step 1}: Over a specified period, collect samples of the time-varying system states, represented as $\{\boldsymbol{w}_t,\boldsymbol{h}_t\mid t\in\mathcal{I}\}$. In this context, $\boldsymbol{w}_t=[w_1^t,....w_N^t]^T$, $\boldsymbol{h}_t = [h_1^t,...,h_N^t]^T$, and $\mathcal{I}$ denotes a set of time indices, and $\mid\mathcal{I}\mid$ indicates the total number of gathered samples.}
\item{\emph{Step 2}: Apply Algorithm 2 to ascertain the optimal number of collaborative clusters for each system state sample. These determined numbers of clusters serve as labels for the system state samples, forming a training dataset denoted as $\{\boldsymbol{w}_t,\boldsymbol{h}_t,m_t^*\mid t\in\mathcal{I}\}$.}
\item{\emph{Step 3}: Build a Deep Neural Network (DNN) with an input layer dimension of $2N$, which aligns with the size of $\{\boldsymbol{w}_{t},\boldsymbol{h}_{t}\}$, and an output layer dimension of $1$, corresponding to $m_{t}^{*}$.}
\item{\emph{Step 4}: Utilize the stochastic gradient descent method to train the DNN with the training dataset $\{\boldsymbol{w}_t,\boldsymbol{h}_t,m_t^*\mid t\in\mathcal{I}\}$. The training process employs the cross-entropy loss function. Before commencing training, it is imperative to standardize and scale the data to adhere to a standard normal distribution with a mean of $0$ and a variance of $1$. This preprocessing step facilitates a swifter training process.}
\item{\emph{Step 5}: Save the trained model and scaling parameters used to normalize the training data. Next, deploy the model to the HAP and import the corresponding parameters.}
\item{\emph{Step 6}: In the decision-making process, the HAP inputs the acquired system state indicators, denoted as $\{\boldsymbol{w}_{new},\boldsymbol{h}_{new}\}$, into the DNN model following an identical normalization procedure. Next, employ Algorithm 3 to ascertain the suitable decision. The particulars of Algorithm 3 are expounded upon below.}
\end{enumerate}

\begin{figure}[t]
\centerline{\includegraphics[width=0.5\textwidth]{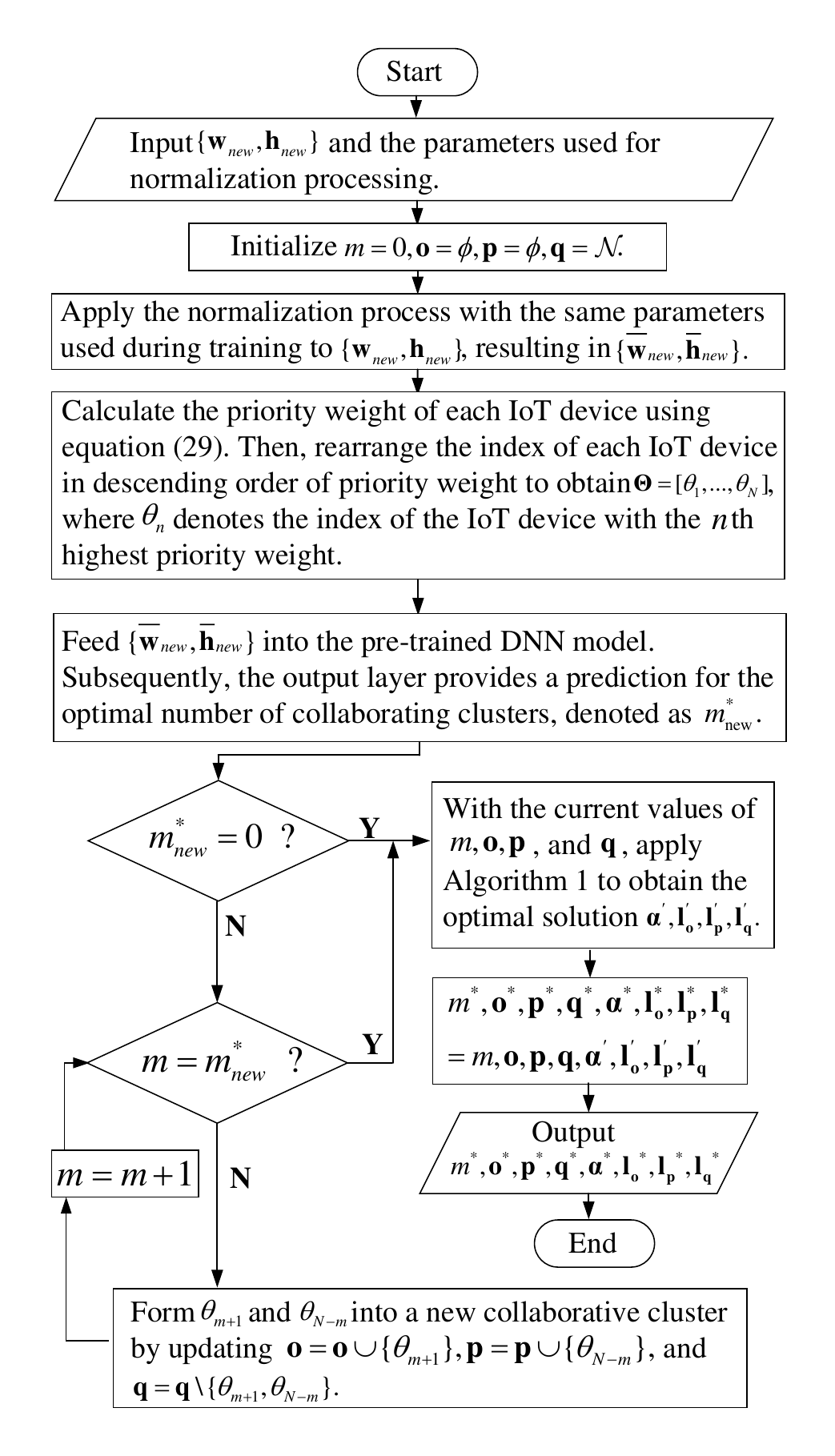}}
\caption{Algorithm 3: DL-Based Algorithm for Solving Problem $\mathcal{P}_2$.}
\label{algo3}
\vspace{-0.5cm}
\end{figure}

By leveraging the trained DNN, the iterative process of finding the optimal number of collaborative clusters is skipped. As a result, the deep learning-based scheme achieves a computational complexity of $\mathcal{O}(LN^{3})$.
\section{Numerical Results}
In this section, some simulation results are provided to verify our theoretical results and evaluate the performance of our proposed algorithms. The network scenario shown in Fig. 1 is simulated. In all simulations, we use the parameters of Powercast TX91501-3W with $P=3$ Watts for the energy transmitter at the HAP, and those of P2110 Powerharvester for the energy receiver at each IoT device\cite{25}. The energy harvesting efficiency is set as $\eta=0.51$, and the average channel gain, denoted as $\overline{h}_n$, adheres to the free-space path loss model outlined below:
\begin{align}
\label{P20}
\overline{h}_n=A_d(\frac{3\times10^8}{4\pi f_cd_n})^{d_e},n\in\mathcal{N},
\end{align}
with $A_d=4.11$ denoting the antenna gain, $f_c=915$ MHz denoting the transmission frequency, and $d_e=2.8$ denoting the path loss exponent. The time-varying wireless channel gain of the $n$th IoT device at time frame $t$ is generated using a Rayleigh fading channel model, expressed as $h_n^t = \overline{h}_n \alpha_n^t$. Here, $\alpha_n^t $ represents the independent random channel fading factor, which follows an exponential distribution with a unit mean.

Without loss of generality, the time block length is set as $T=1$ s, the number of IoT devices is set as $N=20$, and the minimum data processing requirement within a time frame of each IoT device is set as $l_{th}=4000$ bit. For local computing, the computation energy efficiency coefficient and the number of CPU cycles needed to process each bit of raw data are set as $k_{n}=10^{-26}$ and  $\phi_n=100, n\in \mathcal{N}$, respectively \cite{17}. For computing offloading, the bandwidth is set as $B=2$ MHz and the noise power is set as $N_0=10^{-10}$ W. All above parameters will not change unless otherwise stated. Moreover, all the simulations are performed on a Pytorch 1.10.2 platform with an Intel Core i7-13650HX 4.9 GHz CPU and 16 GB of memory.

First, we assess the performance of our two proposed algorithms in comparison to the exhaustive search method. In this method, all  $1+\sum_{m=1}^{m=\left\lfloor\frac N2\right\rfloor}C_N^{2m}(2m-1)!!$ potential scenarios for constructing collaborative clusters are explored to achieve the highest WSCR.

\begin{figure}[t]
\centerline{\includegraphics[width=0.5\textwidth]{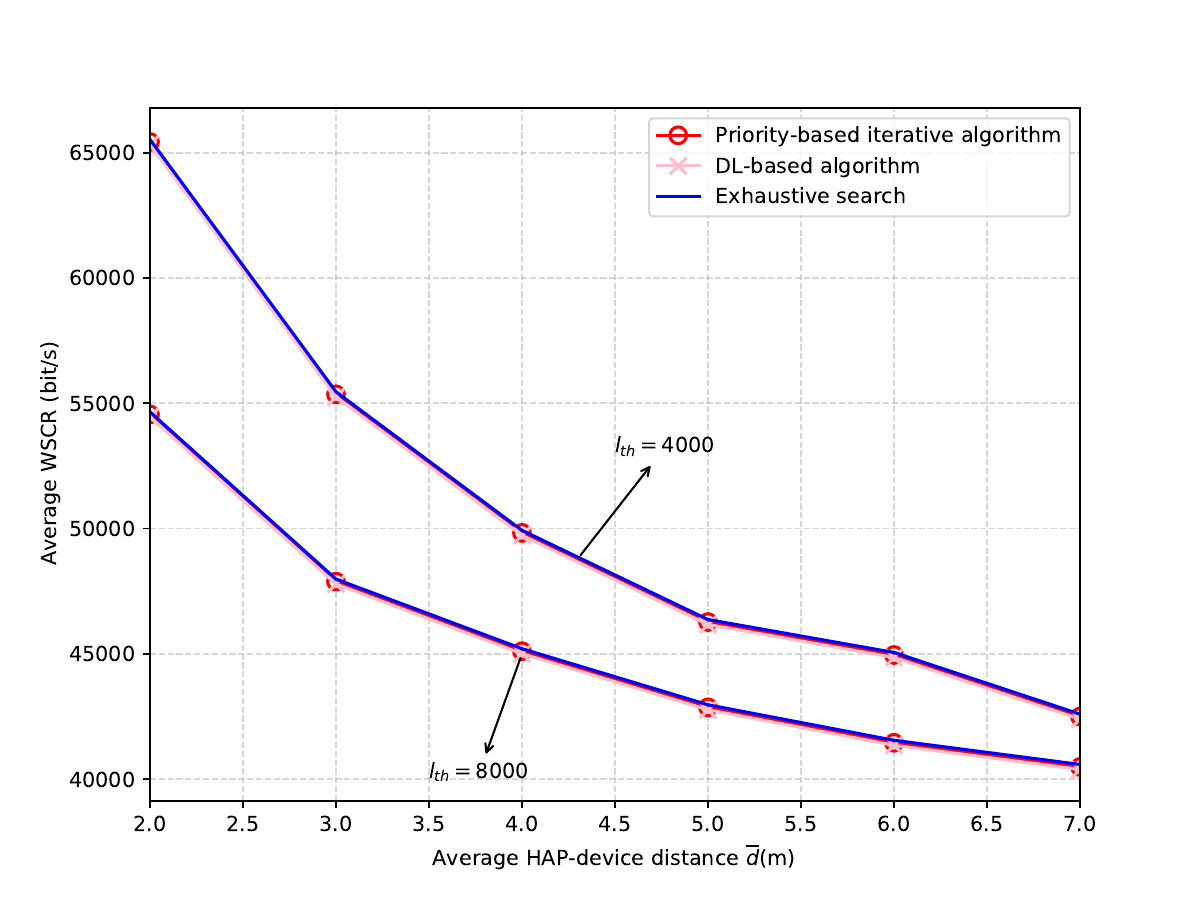}}
\caption{The average WSCR of the system versus the average distance from the IoT devices to the HAP under different minimum data processing requirements.}
\label{WSCR_d}
\end{figure}

Fig. \ref{WSCR_d} shows the achievable average WSCR of the system versus the average distance $\overline{d}$ from the IoT devices to the HAP under different minimum data processing requirements within a time frame, i.e., $l_{th} = 4000$ (the upper curve) and $l_{th} = 8000$ (the lower curve). It is evident that as $\overline{d}$ increases, the achievable average WSCR decreases. This is attributed to the fact that greater distances result in poorer channel conditions. Furthermore, an increase in $l_{th}$ leads to a reduction in the achievable average WSCR for all three algorithms. Notably, the performance of our proposed DL-based algorithm closely aligns with that of the priority-based iterative algorithm and the exhaustive search approach. This suggests that our proposed algorithms have the capacity to converge towards the globally optimal solution.

\begin{figure}[]
\centerline{\includegraphics[width=0.5\textwidth]{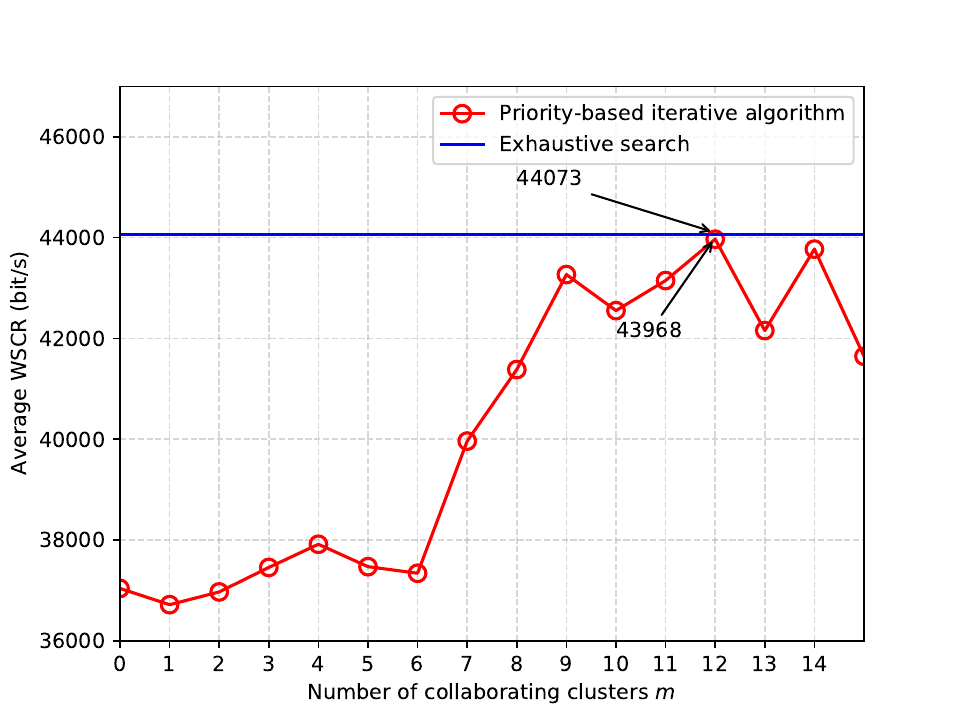}}
\caption{The priority-based iterative algorithm curve versus iteration number.}
\label{WSCR_it}
\end{figure}

In Fig. \ref{WSCR_it}, we analyze the performance of the proposed priority-based iterative algorithm by plotting the evolution of the achievable average WSCR of the system at each iteration. Here, we set $l_{th} = 4000$ bit, $N = 30$, $\overline{d} = 4$ m, while keeping the remaining parameters consistent with Fig. \ref{WSCR_d}. It is notable that the iterative search curve does not exhibit a monotonic or concave pattern characterized by an initial increase followed by a decrease. Additionally, the curve reaches its closest value to the exhaustive search result when the number of collaborating clusters is 12. This behavior stems from the influence of the number of collaborative clusters on the optimal time and data allocation scheme, resulting in irregular fluctuations in the average WSCR with respect to the number of collaborative clusters. As a consequence, the proposed priority-based iterative method entails an exploration of all conceivable scenarios for the number of collaborating clusters, leading to a substantial increase in execution time.

\begin{figure}[]
\centerline{\includegraphics[width=0.5\textwidth]{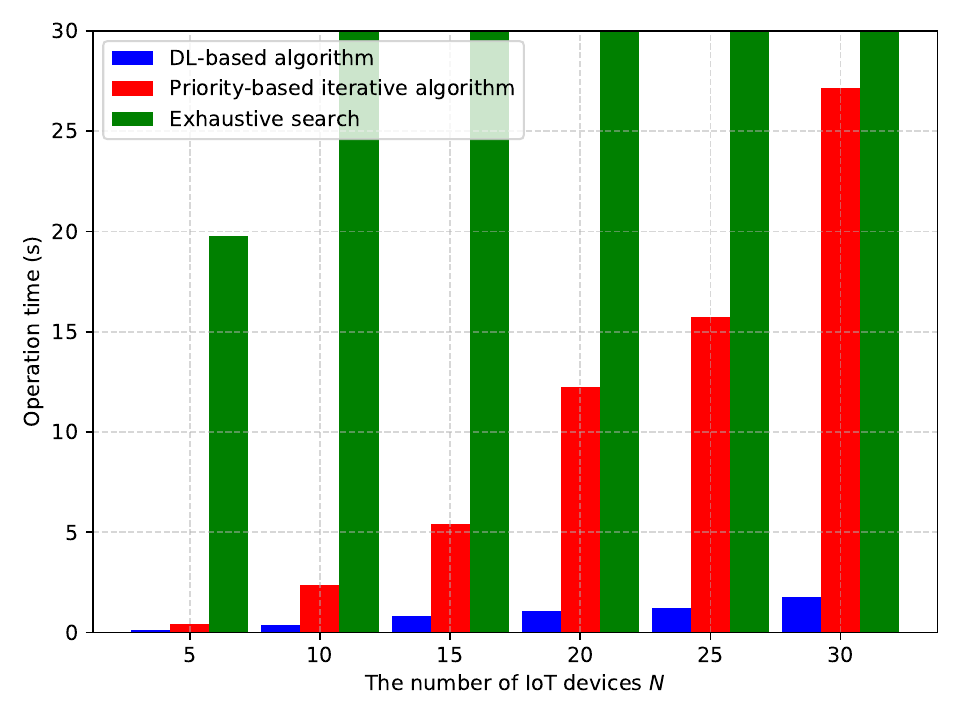}}
\caption{The average operation time versus the number of IoT devices.}
\label{WSCR_t}
\end{figure}

In Fig. \ref{WSCR_t}, the average operation time of the three algorithms is plotted against the number of IoT devices $N$, with the distance from the devices to the HAP being randomly generated. It is noteworthy that for relatively small numbers of IoT devices, e.g., $N=5$, the average operation time required for the priority-based iterative algorithm closely matches that of the DL-based algorithm. However, as $N$ increases, the operation time required for the priority-based iterative algorithm experiences a significant rise. This may be attributed to the fact that an increased number of IoT devices introduces more iterations in the process of finding the optimal number of collaborating clusters. In contrast, the required operation time for the DL-based algorithm does not exhibit a substantial increase with the growth of $N$. Furthermore, the average operation time of the two proposed algorithms is much lower than that of the exhaustive search algorithm. Based on the observations in Fig. \ref{WSCR_t}, it is more advisable to employ the priority-based iterative algorithm in networks with relatively small sizes, while the DL-based algorithm is preferred for large-scale networks. This conclusion aligns with our theoretical analysis of the computational complexities of the two algorithms.

Subsequently, we contrast the proposed two algorithms with the following three additional benchmark schemes. Additionally, we include the result of the exhaustive search method (EX) for comparison.

\begin{enumerate}
\item{\textbf{Local Computing only (LC):} In this scheme, all IoT devices’ data is processed locally.}
\item{\textbf{Non-Collaborative Partial Offloading (NC):} This scheme does not involve collaboration among IoT devices. The entire system follows the time and data allocation strategy derived from Algorithm 1.}
\item{\textbf{Stochastic collaboration (SC):} In this scheme, the collaboration strategy of IoT devices is randomly selected.}
\end{enumerate}

\begin{figure}[]
\centerline{\includegraphics[width=0.5\textwidth]{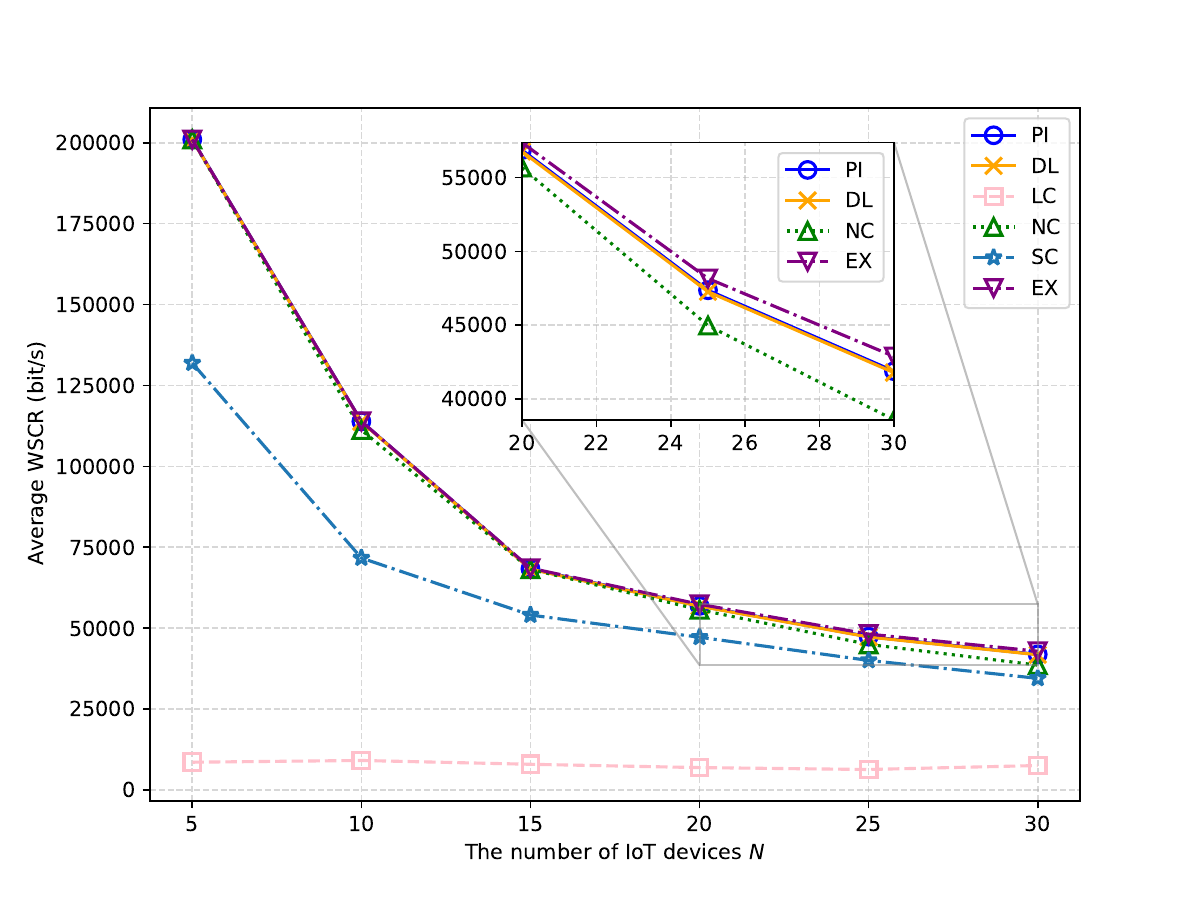}}
\caption{The average WSCR of the system versus the number of IoT devices.}
\label{WSCR_N}
\end{figure}

Fig. \ref{WSCR_N} illustrates the achievable average WSCR of the system versus the number of IoT devices $N$. It is observed that as $N$ increases, the average WSCR decreases for all schemes except the LC strategy. This can be attributed to the gradual decrease in the average computational resources allocated by the HAP to the IoT devices with the increasing number of IoT devices. Additionally, the performance of the SC scheme consistently lags behind that of the NC scheme, underscoring the importance of a meticulously designed collaboration strategy among users. As the number of IoT devices $N$ increases, the gap between the proposed schemes (priority-based iterative algorithm (PI) and DL-based algorithm (DL)) and the NC scheme widens. This underscores the imperative of optimizing collaboration among users in scenarios with a large number of IoT devices, while the benefits of such collaboration are less pronounced in scenarios with a small number of IoT devices. Finally, both of our proposed schemes demonstrate performance that is highly comparable to that of the exhaustive search algorithm, indicating their capability to achieve nearly optimal performance.

\begin{figure}[]
\centerline{\includegraphics[width=0.5\textwidth]{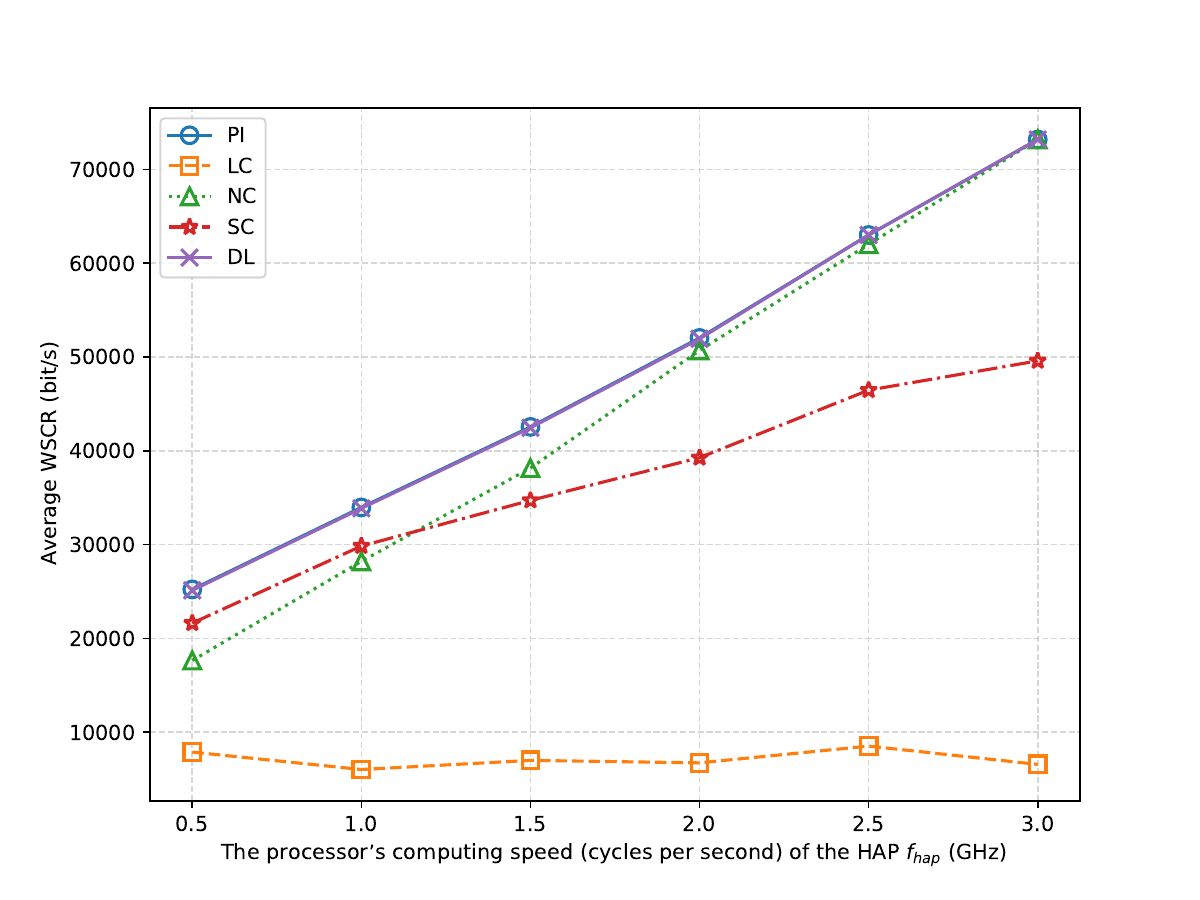}}
\caption{The average WSCR of the system versus the processor’s computing speed of the HAP.}
\label{WSCR_f}
\end{figure}

Fig. \ref{WSCR_f} illustrates the average WSCR of the system versus the processor’s computing speed of the HAP $f_{hap}$. It is evident that as $f_{hap}$ increases, the average WSCR increases for all schemes except the LC strategy. When $f_{hap} > 1$ GHz, the NC scheme surpasses the SC scheme in performance. Additionally, with the increase of $f_{hap}$, the performance of the NC scheme gradually converges towards that of the proposed schemes. This emphasizes that collaboration among devices becomes less crucial when the computational resources of the HAP are ample. Conversely, when $f_{ap} < 1$GHz, the SC scheme outperforms the NC scheme.

\begin{figure}[]
\centerline{\includegraphics[width=0.5\textwidth]{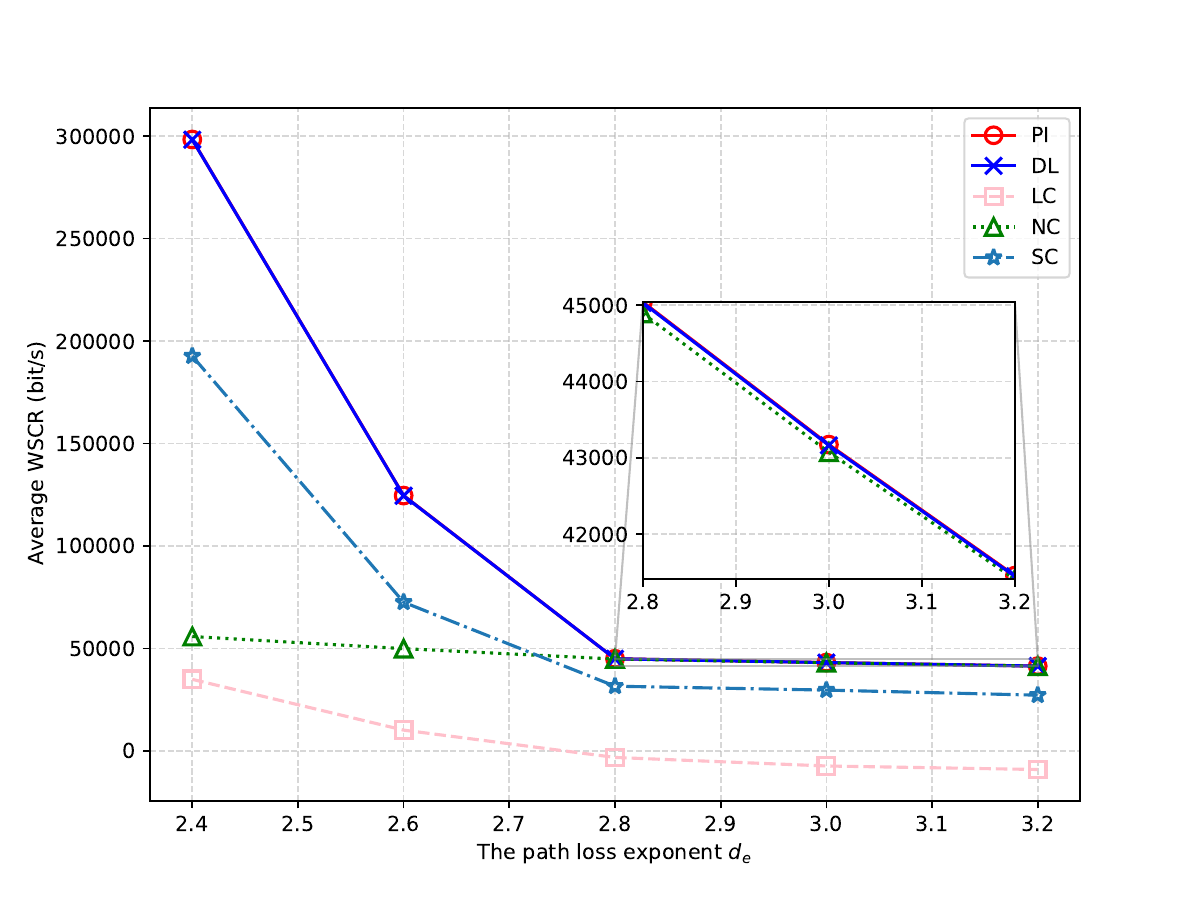}}
\caption{The average WSCR of the system versus the path loss exponent $d_e$.}
\label{WSCR_de}
\end{figure}

Fig. \ref{WSCR_de} displays the average WSCR of the system with respect to the path loss exponent $d_e$. It is evident that with an increase in $d_e$, the average WSCR of all schemes experiences a significant decrease. When $d_e<2.7$, the SC scheme outperforms the NC scheme. However, for $d_e \ge 2.8$, the NC scheme surpasses the SC scheme and gradually approaches the performance level of the proposed schemes. This may be attributed to the fact that, as the channel quality deteriorates, the energy harvested by the devices decreases, while the energy required for offloading increases. Consequently, the performance enhancement resulting from user collaboration gradually diminishes. Therefore, collaboration among users becomes unnecessary when the channel quality is poor.

Lastly, we examine the influence of various parameters on device collaboration. In particular, we plot the variation in the optimal number of collaborating clusters obtained through the priority-based iterative algorithm under different parameters.

\begin{figure}[]
\centerline{\includegraphics[width=0.5\textwidth]{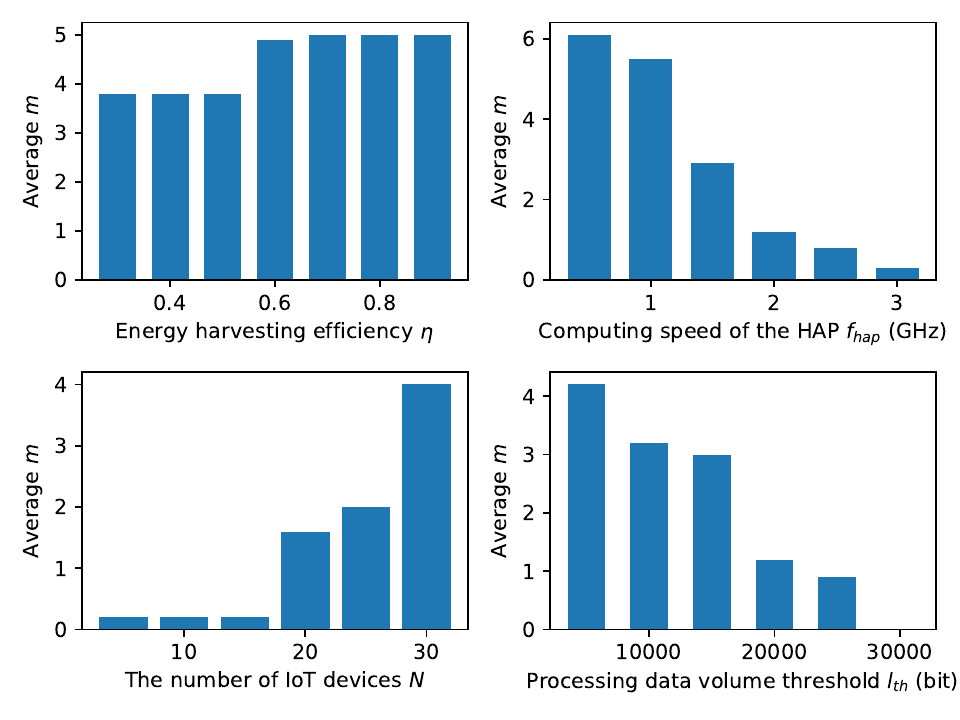}}
\caption{The optimal number of collaborating clusters $m$ obtained by the priority-based iterative algorithm versus various parameters.}
\label{WSCR_num}
\end{figure}

Fig. \ref{WSCR_num} presents the optimal number of collaborating clusters $m$ obtained through the priority-based iterative algorithm versus various parameters. It is evident that with an increase in the energy harvesting efficiency $\eta$ or the number of IoT devices $N$, the optimal number of collaborative clusters also increases. Conversely, as the processor’s computing speed of the HAP $f_{hap}$ or the minimum data processing requirement within a time frame for each device increases, the optimal number of collaborative clusters decreases. Consequently, in scenarios where energy is the primary limiting factor, collaboration among devices is restrained, while in scenarios where the computational capacity of the HAP is the main limiting factor, collaboration among users is encouraged. Furthermore, an increase in the minimum data processing requirement within a time frame for each device hinders collaboration among users.

Based on the simulation results above, it is evident that the DL-based algorithm achieves near-optimal performance within a short runtime. Hence, the proposed scheme exhibits significant feasibility and advantages for real-time decision-making in real-world fading environments. Nevertheless, the scheme possesses potential limitations, such as inadequate adaptation to IoT devices joining or exiting the network. To mitigate this issue, in practical deployment, DNN models corresponding to different numbers of IoT devices can be pre-trained and switched based on the varying number of IoT devices. However, this approach necessitates the upfront collection of extensive data for different IoT device counts, incurring high costs. In contrast, an alternative scheme employs the priority-based iterative algorithm for initial decision-making, saving samples and labels for subsequent use. Upon accumulating sufficient training data for a specific number of IoT devices, the corresponding DNN model is trained. Subsequently, when encountering the specified number of IoT devices, the trained DNN is utilized. Over time, DNN adoption will become the norm.

\section{Conclusion}
This paper investigates the design for maximizing the weighted sum computation rate in a wireless-powered MEC network with multi-user cooperation. A mathematical optimization problem is formulated by jointly optimizing collaboration, time allocation, and data distribution among multiple IoT devices and the HAP while adhering to constraints on energy causality and devices' computing rate requirements. To efficiently address this mixed-integer programming problem, an interior-point method based algorithm and a priority-based iterative algorithm are initially proposed. Then, a deep learning-based approach is introduced to accelerate the algorithm's operation further. Simulation results demonstrate that the proposed algorithms achieve comparable performance with the exhaustive search method. Furthermore, the deep learning-based approach significantly reduces the algorithm's execution time compared to the priority-based iterative algorithm and the exhaustive search method.

\bibliographystyle{IEEEtran}
\bibliography{main}

\begin{thebibliography}{10}
\providecommand{\url}[1]{#1}
\csname url@samestyle\endcsname
\providecommand{\newblock}{\relax}
\providecommand{\bibinfo}[2]{#2}
\providecommand{\BIBentrySTDinterwordspacing}{\spaceskip=0pt\relax}
\providecommand{\BIBentryALTinterwordstretchfactor}{4}
\providecommand{\BIBentryALTinterwordspacing}{\spaceskip=\fontdimen2\font plus
\BIBentryALTinterwordstretchfactor\fontdimen3\font minus \fontdimen4\font\relax}
\providecommand{\BIBforeignlanguage}[2]{{%
\expandafter\ifx\csname l@#1\endcsname\relax
\typeout{** WARNING: IEEEtran.bst: No hyphenation pattern has been}%
\typeout{** loaded for the language `#1'. Using the pattern for}%
\typeout{** the default language instead.}%
\else
\language=\csname l@#1\endcsname
\fi
#2}}
\providecommand{\BIBdecl}{\relax}
\BIBdecl

\bibitem{1}
L.~Da~Xu, W.~He, and S.~Li, ``Internet of things in industries: {A} survey,'' \emph{IEEE Trans. Ind. Informat.}, vol.~10, no.~4, pp. 2233--2243, Nov. 2014.

\bibitem{2}
I.~Bisio, C.~Garibotto, A.~Grattarola, F.~Lavagetto, and A.~Sciarrone, ``{Exploiting context-aware capabilities over the Internet of Things for Industry 4.0 applications},'' \emph{IEEE Netw.}, vol.~32, no.~3, pp. 101--107, May/Jun. 2018.

\bibitem{3}
E.~Sisinni, A.~Saifullah, S.~Han, U.~Jennehag, and M.~Gidlund, ``{Industrial Internet of Things: Challenges, opportunities, and directions},'' \emph{IEEE Trans. Ind. Informat.}, vol.~14, no.~11, pp. 4724--4734, Nov. 2018.

\bibitem{4}
M.~Aazam, S.~Zeadally, and K.~A. Harras, ``Deploying fog computing in industrial internet of things and industry 4.0,'' \emph{IEEE Trans. Ind. Informat.}, vol.~14, no.~10, pp. 4674--4682, Oct. 2018.

\bibitem{5}
K.~Kaur, S.~Garg, G.~S. Aujla, N.~Kumar, J.~J. Rodrigues, and M.~Guizani, ``Edge computing in the industrial internet of things environment: Software-defined-networks-based edge-cloud interplay,'' \emph{IEEE Commun. Mag.}, vol.~56, no.~2, pp. 44--51, Feb. 2018.

\bibitem{42}
S.~Wang, X.~Zhang, Y.~Zhang, L.~Wang, J.~Yang, and W.~Wang, ``{A survey on mobile edge networks: Convergence of computing, caching and communications},'' \emph{IEEE Access}, vol.~5, pp. 6757--6779, Mar. 2017.

\bibitem{43}
Y.~Sun, B.~Lei, J.~Liu, H.~Huang, X.~Zhang, J.~Peng, and W.~Wang, ``{Computing power network: A survey},'' \emph{arXiv preprint arXiv:2210.06080}, Nov. 2022.

\bibitem{44}
Y.~Li, X.~Ge, B.~Lei, X.~Zhang, and W.~Wang, ``Joint task partitioning and parallel scheduling in device-assisted mobile edge networks,'' \emph{IEEE Internet Things J.}, Early Access, Dec. 2023.

\bibitem{45}
J.~Liu, Y.~Sun, J.~Su, Z.~Li, X.~Zhang, B.~Lei, and W.~Wang, ``{Computing power network: A testbed and applications with edge intelligence},'' in \emph{IEEE INFOCOM 2022-IEEE Conference on Computer Communications Workshops (INFOCOM WKSHPS)}, 2022, pp. 1--2.

\bibitem{46}
Y.~Li, B.~Lei, Z.~Li, Z.~Qu, X.~Zhang, and W.~Wang, ``Task offloading with multi-cluster collaboration for computing and network convergence,'' in \emph{Proceedings of the 29th Annual International Conference on Mobile Computing and Networking}, 2023, pp. 1--3.

\bibitem{6}
F.~Zhou and R.~Q. Hu, ``Computation efficiency maximization in wireless-powered mobile edge computing networks,'' \emph{IEEE Trans. Wireless Commun.}, vol.~19, no.~5, pp. 3170--3184, May. 2020.

\bibitem{27}
Y.~Wang, M.~Sheng, X.~Wang, L.~Wang, and J.~Li, ``{Mobile-edge computing: Partial computation offloading using dynamic voltage scaling},'' \emph{IEEE Trans. Commun.}, vol.~64, no.~10, pp. 4268--4282, Oct. 2016.

\bibitem{13}
Y.~Li, X.~Zhang, Y.~Sun, J.~Liu, B.~Lei, and W.~Wang, ``{Joint offloading and resource allocation with partial information for multi-user edge computing},'' in \emph{Proc. IEEE Globecom Workshops (GC Wkshps)}, Rio de Janeiro, Brazil, 2022, pp. 1736--1741.

\bibitem{29}
Y.~Mao, J.~Zhang, and K.~B. Letaief, ``Dynamic computation offloading for mobile-edge computing with energy harvesting devices,'' \emph{IEEE J. Sel. Areas Commun.}, vol.~34, no.~12, pp. 3590--3605, Dec. 2016.

\bibitem{30}
J.~Liu, Y.~Mao, J.~Zhang, and K.~B. Letaief, ``{Delay-optimal computation task scheduling for mobile-edge computing systems},'' in \emph{Proc. IEEE Int. Symp. Inf. Theory}, Barcelona, Spain, 2016, pp. 1451--1455.

\bibitem{8}
J.~Liu, K.~Xiong, D.~W.~K. Ng, P.~Fan, Z.~Zhong, and K.~B. Letaief, ``Max-min energy balance in wireless-powered hierarchical fog-cloud computing networks,'' \emph{IEEE Trans. Wireless Commun.}, vol.~19, no.~11, pp. 7064--7080, 2020.

\bibitem{9}
R.~Jiang, K.~Xiong, P.~Fan, Y.~Zhang, and Z.~Zhong, ``Power minimization in swipt networks with coexisting power-splitting and time-switching users under nonlinear {EH} model,'' \emph{IEEE Internet Things J.}, vol.~6, no.~5, pp. 8853--8869, Oct. 2019.

\bibitem{10}
X.~Lu, P.~Wang, D.~Niyato, D.~I. Kim, and Z.~Han, ``{Wireless charging technologies: Fundamentals, standards, and network applications},'' \emph{IEEE Commun. Surveys Tut.}, vol.~18, no.~2, pp. 1413--1452, Apr.–Jun. 2015.

\bibitem{11}
F.~Wang, J.~Xu, and S.~Cui, ``Optimal energy allocation and task offloading policy for wireless powered mobile edge computing systems,'' \emph{IEEE Trans. Wireless Commun.}, vol.~19, no.~4, pp. 2443--2459, Apr. 2020.

\bibitem{12}
T.~Campi, S.~Cruciani, F.~Maradei, and M.~Feliziani, ``Coil design of a wireless power-transfer receiver integrated into a left ventricular assist device,'' \emph{Electronics}, vol.~10, no.~8, p. 874, Apr. 2021.

\bibitem{15}
F.~Wang, J.~Xu, and S.~Cui, ``Optimal energy allocation and task offloading policy for wireless powered mobile edge computing systems,'' \emph{IEEE Trans. Wireless Commun.}, vol.~19, no.~4, pp. 2443--2459, Apr. 2020.

\bibitem{16}
C.~Psomas and I.~Krikidis, ``Wireless powered mobile edge computing: Offloading or local computation?'' \emph{IEEE Commun. Lett.}, vol.~24, no.~11, pp. 2642--2646, Nov. 2020.

\bibitem{17}
H.~Li, K.~Xiong, Y.~Lu, B.~Gao, P.~Fan, and K.~B. Letaief, ``Distributed design of wireless powered fog computing networks with binary computation offloading,'' \emph{IEEE Trans. Mobile Comput.}, vol.~22, no.~4, pp. 2084--2099, Apr. 2023.

\bibitem{18}
J.~Liu, K.~Xiong, D.~W.~K. Ng, P.~Fan, Z.~Zhong, and K.~B. Letaief, ``Max-min energy balance in wireless-powered hierarchical fog-cloud computing networks,'' \emph{IEEE Trans. Wireless Commun.}, vol.~19, no.~11, pp. 7064--7080, Nov. 2020.

\bibitem{19}
K.~Xiong, Y.~Liu, L.~Zhang, B.~Gao, J.~Cao, P.~Fan, and K.~B. Letaief, ``Joint optimization of trajectory, task offloading, and {CPU} control in {UAV}-assisted wireless powered fog computing networks,'' \emph{IEEE Trans. Green Commun. Netw.}, vol.~6, no.~3, pp. 1833--1845, Sep. 2022.

\bibitem{20}
P.~X. Nguyen, D.-H. Tran, O.~Onireti, P.~T. Tin, S.~Q. Nguyen, S.~Chatzinotas, and H.~V. Poor, ``Backscatter-assisted data offloading in {OFDMA}-based wireless-powered mobile edge computing for {IoT} networks,'' \emph{IEEE Internet Things J.}, vol.~8, no.~11, pp. 9233--9243, Jun. 2021.

\bibitem{21}
F.~Zhou and R.~Q. Hu, ``Computation efficiency maximization in wireless-powered mobile edge computing networks,'' \emph{IEEE Trans. Wireless Commun.}, vol.~19, no.~5, pp. 3170--3184, May. 2020.

\bibitem{22}
B.~Su, Q.~Ni, W.~Yu, and H.~Pervaiz, ``Optimizing computation efficiency for {NOMA}-assisted mobile edge computing with user cooperation,'' \emph{IEEE Trans. Green Commun. Netw.}, vol.~5, no.~2, pp. 858--867, Jun. 2021.

\bibitem{23}
B.~Li, F.~Si, W.~Zhao, and H.~Zhang, ``Wireless powered mobile edge computing with {NOMA} and user cooperation,'' \emph{IEEE Trans. Veh. Technol.}, vol.~70, no.~2, pp. 1957--1961, Feb. 2021.

\bibitem{24}
X.~Wu, Y.~He, and A.~Saleem, ``Computation rate maximization in multi-user cooperation-assisted wireless-powered mobile edge computing with {OFDMA},'' \emph{China Commun}, vol.~20, no.~1, pp. 218--229, Jan. 2023.

\bibitem{14}
C.~You, K.~Huang, H.~Chae, and B.-H. Kim, ``Energy-efficient resource allocation for mobile-edge computation offloading,'' \emph{IEEE Trans. Wireless Commun.}, vol.~16, no.~3, pp. 1397--1411, Mar. 2016.

\bibitem{40}
Q.~V. Khanh, A.~Chehri, N.~M. Quy, N.~D. Han, and N.~T. Ban, ``{Innovative trends in the 6G era: A comprehensive survey of architecture, applications, technologies, and challenges},'' \emph{IEEE Access}, Apr. 2023.

\bibitem{39}
S.~Herbert, I.~Wassell, T.-H. Loh, and J.~Rigelsford, ``Characterizing the spectral properties and time variation of the in-vehicle wireless communication channel,'' \emph{IEEE Trans. Commun.}, vol.~62, no.~7, pp. 2390--2399, Jul. 2014.

\bibitem{36}
G.~Chen, Y.~Chen, Z.~Mai, C.~Hao, M.~Yang, and L.~Du, ``{Incentive-based distributed resource allocation for task offloading and collaborative computing in MEC-enabled networks},'' \emph{IEEE Internet Things J.}, vol.~10, no.~10, pp. 9077--9091, May 2023.

\bibitem{37}
T.~Bai, C.~Pan, Y.~Deng, M.~Elkashlan, A.~Nallanathan, and L.~Hanzo, ``Latency minimization for intelligent reflecting surface aided mobile edge computing,'' \emph{IEEE J. Sel. Areas Commun.}, vol.~38, no.~11, pp. 2666--2682, Nov. 2020.

\bibitem{31}
S.~Bi and Y.~J. Zhang, ``Computation rate maximization for wireless powered mobile-edge computing with binary computation offloading,'' \emph{IEEE Trans. Wireless Commun.}, vol.~17, no.~6, pp. 4177--4190, 2018.

\bibitem{32}
M.~Zeng, R.~Du, V.~Fodor, and C.~Fischione, ``{Computation rate maximization for wireless powered mobile edge computing with NOMA},'' in \emph{Proc. IEEE Int. Symp. World Wireless, Mobile Multimedia Netw.}, Washington, DC, USA, 2019, pp. 1--9.

\bibitem{33}
Y.~Lu, K.~Xiong, P.~Fan, Z.~Zhong, and K.~B. Letaief, ``{Robust transmit beamforming with artificial redundant signals for secure SWIPT system under non-linear EH model},'' \emph{IEEE Trans. Wireless Commun.}, vol.~17, no.~4, pp. 2218--2232, Apr. 2018.

\bibitem{34}
E.~Boshkovska, D.~W.~K. Ng, N.~Zlatanov, A.~Koelpin, and R.~Schober, ``{Robust resource allocation for MIMO wireless powered communication networks based on a non-linear EH model},'' \emph{IEEE Trans. Commun.}, vol.~65, no.~5, pp. 1984--1999, 2017.

\bibitem{35}
R.~Morsi, E.~Boshkovska, E.~Ramadan, D.~W.~K. Ng, and R.~Schober, ``On the performance of wireless powered communication with non-linear energy harvesting,'' in \emph{Proc. IEEE Int. Workshop Signal Process. Adv. Wireless Commun.}, Sapporo, Japan, 2017, pp. 1--5.

\bibitem{25}
L.~Huang, S.~Bi, and Y.-J.~A. Zhang, ``Deep reinforcement learning for online computation offloading in wireless powered mobile-edge computing networks,'' \emph{IEEE Trans. Mobile Comput.}, vol.~19, no.~11, pp. 2581--2593, Nov. 2019.

\bibitem{38}
F.~Wang, J.~Xu, X.~Wang, and S.~Cui, ``Joint offloading and computing optimization in wireless powered mobile-edge computing systems,'' \emph{IEEE Trans. Wireless Commun.}, vol.~17, no.~3, pp. 1784--1797, 2017.

\bibitem{41}
L.~Xiao and P.~Li, ``Improvement on mean shift based tracking using second-order information,'' in \emph{2008 19th International Conference on Pattern Recognition}, Tampa, FL, USA, 2008, pp. 1--4.

\end{thebibliography}

\end{document}